\documentclass[10pt]{article}

\usepackage{amsmath}
\usepackage{amssymb}
\usepackage{amsfonts}
\usepackage{amsthm}
\usepackage[english]{babel}
\usepackage{graphicx}
\usepackage[all]{xy}
\usepackage{amssymb}
\usepackage{hyperref}

\usepackage[OT2,OT1]{fontenc}
\newcommand\cyr{%
\renewcommand\rmdefault{wncyr}%
\renewcommand\sfdefault{wncyss}%
\renewcommand\encodingdefault{OT2}%
\normalfont
\selectfont}
\DeclareTextFontCommand{\textcyr}{\cyr}

\newtheorem{theorem}{Theorem}[section]
\newtheorem{lemma}[theorem]{Lemma}
\newtheorem{definition}[theorem]{Definition}

\newtheorem{proposition}[theorem]{Proposition}
\newtheorem{remark}[theorem]{Remark}

\newcommand{\mR}{\mathbb{R}}
\newcommand{\mC}{\mathbb{C}}
\newcommand{\mN}{\mathbb{N}}
\newcommand{\mE}{\mathbb{E}}
\newcommand{\mZ}{\mathbb{Z}}
\newcommand{\mS}{\mathbb{S}}

\newcommand{\mH}{\mathbb{H}}

\newcommand{\cF}{\mathcal{F}}

\newcommand{\cC}{\mathcal{C}}

\newcommand{\cS}{\mathcal{S}}

\newcommand{\cA}{\mathcal{A}}

\newcommand{\fu}{\mathfrak{u}}
\newcommand{\fv}{\mathfrak{v}}

\newcommand{\fl}{\mathfrak{l}}

\newcommand{\tr}{{\rm{tr}\,}}

\newcommand{\End}{{\rm{End}}}
\newcommand{\Span}{{\rm{Span}}}
\newcommand{\rmU}{{\rm U\,}}
\newcommand{\rmUSp}{{\rm USp\,}}
\newcommand{\SO}{{\rm SO\,}}
\newcommand{\St}{{\rm St\,}}
\newcommand{\Gl}{{\rm Gl\,}}
\newcommand{\rmUn}{{\mathfrak u\,}}
\newcommand{\rmUsp}{{\mathfrak{usp}\,}}
\newcommand{\so}{{\mathfrak{so}\,}}
\newcommand{\gl}{{\mathfrak{gl}\,}}
\newcommand{\Herm}{{\rm Herm\,}}
\newcommand{\Pf}{{\rm Pf\,}}
\newcommand{\sign}{{\rm sign\,}}
\newcommand{\eins}{\leavevmode\hbox{\small1\kern-3.8pt\normalsize1}}

\newcommand{\Li}{\mathfrak{g}^{(\beta)}_n}
\newcommand{\Lii}{\mathfrak{g}^{(\beta)}}
\newcommand{\FS}{\cA^{(\beta)}(n,m)}

\begin{document}
\title{Pizzetti formulae for Stiefel manifolds and applications}

\author{K.\ Coulembier $\quad$ M.\ Kieburg}

\maketitle

\begin{abstract}
Pizzetti's formula explicitly shows the equivalence of the rotation invariant integration over a sphere and the action of rotation invariant differential operators. We generalize this idea to the integrals over real, complex, and quaternion Stiefel manifolds in a unifying way. In particular we propose a new way to calculate group integrals and try to uncover some algebraic structures which manifest themselves for some well-known cases like the Harish-Chandra integral. We apply a particular case of our formula to an Itzykson-Zuber integral for the coset $\SO(4)/[\SO(2)\times\SO(2)]$. This integral naturally appears in the calculation of the two-point correlation function in the transition of the statistics of the Poisson ensemble and the Gaussian orthogonal ensemble in random matrix theory.
\end{abstract}


\section{Introduction}\label{intro}

The calculation of integrals with a group invariant measure is often a difficult undertaking. Such invariant integrals over (cosets of) groups regularly appear in harmonic analysis \cite{Harish-Chandra}, representation theory~\cite{HPZ}, combinatorics~\cite{Louck}, random matrix theory~\cite{Peterbook,SchlittgenWettig,Zirnbauer96}, quantum field theory~\cite{KVZ14,LeutwylerSmilga,ShuryakVerbaarschot,Verbaarschot,VerbaarschotWettig,VerbaarschotZahed}, and many other fields in mathematics, physics and beyond. The unique invariant measure on a compact Lie group is known as the Haar measure and we employ the same name for the induced measure on cosets.

Indeed many approaches were tried to calculate integrals of various functions over the three compact groups $\SO(n)$, $\rmU(n)$, and $\rmUSp(2n)$ with respect to the Haar measure. For instance, very explicit but also quite cumbersome formulae were derived for the integrals over an arbitrary product of matrix elements $U_{ij}$ of a group element $U=\{U_{ij}\}\in\SO(n),\rmU(n),\rmUSp(2n)$, see Ref.~\cite{MR1902484}. More compact results are desirable for some applications. Particularly, when one has to deal with some additional integrals, compact results are more suitable and may show some convenient algebraic structures, such as determinants or Pfaffian determinants. Other approaches tried to uncover these algebraic structures. Some of these approaches were successfully applied to, for example, the Harish-Chandra integral~\cite{Harish-Chandra} and partially the Leutwyler-Smilga integral, see Refs.~\cite{KVZ14,LeutwylerSmilga,Verbaarschot} and subsection~\ref{newPiz} and appendix~\ref{app1} of the current paper. Other integrals as the Itzykson-Zuber integral~\cite{ItzZub} and the Berezin-Karpelevich integral~\cite{BerKarp,SchlittgenWettig}  were only successfully calculated for the unitary group $\rmU(n)$ while these integrals for the groups $\SO(n)$ and $\rmUSp(2n)$ are much harder to evaluate. For the real and quaternion Itzykson-Zuber integral, only recursion formulae~\cite{BergereEynard,MR1893696} and expansions in Jack-polynomials~\cite{Muirhead,OkoOls} were developed. With the help of these approaches group integrals with a small matrix dimension $n$ were solved. However these approaches were never completely successfully applied to the general case of arbitrary matrix dimension.

We pursue a new idea to investigate general group integrals and their hidden algebraic structures. Inspired by Pizzetti's formula~\cite{Pizzetti} for the integral over the unit sphere, we want to replace the integral with the Haar measure by an application of a group invariant differential operator on the integrand. This is natural in the light of the Fourier transform which interchanges between coordinates and partial derivatives. The idea behind Pizzetti's formula has recently also been applied to find formulae for invariant integration in superspace, see \cite{OSpHarm, DBE1}. Our approach will shed some light on the Haar measures. Moreover we hope to find new insights which may help to derive new compact results for some group integrals frequently appearing in physics and mathematics. 

In many physical applications the group invariance of the integrand is not completely broken, meaning that it remains invariant under the action of some subgroup. Hence integrals over cosets are of particular interest, since they are integrals over a smaller space and, thus, over less integration variables. Integrals for the case where the cosets are Stiefel manifolds, that is $\SO(n)/\SO(n-k)$, $\rmU(n)/\rmU(n-k)$ and $\rmUSp(2n)/\rmUSp(2n-2k)$, are the focus of this work. This is motivated by several appearances of such integrals in physics, of which we present some in subsection~\ref{appl-mot}. By our concept of Pizzetti formulae, we map the invariant integral over a large group with matrix dimension $n$ to a function depending on a small number $k$ of group invariant differential operators. This number $k$ is independent of the matrix dimension $n$ and may be more suitable for further calculation or large $n$ analysis. In the extreme case $k=1$, this idea recovers the classical Pizzetti formula, which writes the integral over the unit sphere $\mS^{n-1}$ in terms of the Laplacian on $\mR^n$.

After introducing our notation and some definitions in section~\ref{secprel} we sketch the main idea behind our approach in section~\ref{sketch}. We do not claim any rigorous proofs in this section but it shall give educated guesses, collects the results for all three groups considered and demonstrate how our results can be applied in random matrix theory. Many of these `conjectures' will then be rigorously proven in sections~\ref{secinv} and \ref{sec2}, others will remain conjectures since no direct proofs are found, yet. In particular we show in section~\ref{sec;Piz-beta1} that for any of the three groups  there is a function which, when evaluated in a finite number of invariant differential operators, is equivalent to the group integrals over $\SO(n)$, $\rmU(n)$, and $\rmUSp(2n)$, for a large class of integrands. However, this function will not be explicitly calculated in full generality. We only derive explicit formulae for the case of the Stiefel manifolds $\SO(n)/\SO(n-2)$, $\rmU(n)/\rmU(n-2)$, and $\rmUSp(2n)/\rmUSp(2n-4)$ in appendix~\ref{app0} and rigorously prove them in section \ref{sec2}. 

The other Pizzetti formulae for the Stiefel manifolds $\rmU(n)/\rmU(n-k)$ and $\rmUSp(2n)/\rmUSp(2n-2k)$, presented and derived in section~\ref{sketch} and appendix~\ref{app1}, will remain conjectures. Contrary to the type of formulae in the previous paragraph, these are only implicitly expressed in terms of a finite amount of invariant differential operators, but they do exhibit intriguing determinantal and Pfaffian determinantal structures which frequently appear in random matrix theory~\cite{Gernotbook,Peterbook,Mehtabook}. 

In section~\ref{appl} we apply our result in section~\ref{sec2} for the Stiefel manifold $\SO(n)/\SO(n-2)$ to calculate the Itzykson-Zuber integral corresponding to the coset $\SO(4)/[\SO(2)\times\SO(2)]$. This integral plays an important role in random matrix theory. Especially it is crucially related to the correlation function of two eigenvalues of real symmetric random matrices~\cite{GuhrKohler-b,Heiner-thesis}. The lack of knowledge of this integral prevents exact analytical results for the transition of level statistics of the Poisson ensemble (diagonal matrices with independently identically distributed entries) and the Gaussian orthogonal ensemble (real symmetric matrices with independently identically Gaussian distributed entries)~\cite{GM-GW,Haakebook}. Since our result is more compact and explicit than any other expression known~\cite{BergereEynard,MR1893696,GuhrKohler-b,Heiner-thesis}, we hope that it can contribute to the calculation of the level statistics of this transition ensemble.

The approach in sections \ref{secinv} and \ref{sec2} relies on classical invariant theory. In section \ref{secinv} we also demonstrate how a Howe dual pair, attached to every Stiefel manifold, naturally arises in this context. Furthermore we derive two alternative interpretations of our formulae, within the theory of these Howe dualities. We illustrate this explicitly for $\beta=1$, which makes the connection to harmonic analysis in several sets of variables.

In section~\ref{sec:Schwartz} we make contact to the known recursion formulae~\cite{MR1893696,BergereEynard} by marrying Pizzetti's idea with the idea of splitting the three compact groups and their corresponding Stiefel manifolds in spheres. 

In appendix \ref{prop} we rigorously prove a Pizzetti formula for a certain class of invariant integrals over manifolds, which as three special cases implies the formulae suggested by appendix \ref{app0}.

In section~\ref{conclusio} we present an overview of the main results. 

\section{Preliminaries}
\label{secprel}

In this section we fix some notation and recall some necessary concepts. We consider the three compact connected Lie groups
\begin{equation}
\rmU^{(\beta)}(n) := \left\{\begin{array}{cl} \SO(n), & \beta=1, \\ \rmU(n), & \beta=2, \\ \rmUSp(2n), & \beta=4,\end{array}\right.
\label{groups}
\end{equation}
which are the special orthogonal, the unitary, and the unitary symplectic group, respectively. The cosets we are interested in are the Stiefel manifolds
\begin{equation}
\St^{(\beta)}(n,m) := \rmU^{(\beta)}(n)/\rmU^{(\beta)}(m) =\left\{\begin{array}{cl} \SO(n)/\SO(m), & \beta=1, \\ \rmU(n)/\rmU(m), & \beta=2, \\ \rmUSp(2n)/\rmUSp(2m), & \beta=4,\end{array}\right.
\label{Stiefel}
\end{equation}
where $0\leq m< n\in\mathbb{N}$ is an integer. The Dyson index, $\beta=1,2,4$, reflects the underlying number field which is real ($\beta=1$), complex ($\beta=2$), or quaternion ($\beta=4$). We also set
\begin{eqnarray}\label{gammadef}
 \gamma&:=&\left\{\begin{array}{cl} 1, & \beta=1,2, \\ 2, & \beta=4,\end{array}\right.
\end{eqnarray}
which is convenient when writing our results in a unified way.

We denote the field of quaternion numbers by $\mathbb{H}$ and consider an injective algebra morphism $\mH\hookrightarrow \mC^{2\times 2}$ in terms of Pauli matrices. This means that $(a,b,c,d)\in\mH$ is represented by
\begin{equation}\label{quat-Pauli}
\left(\begin{array}{cc}a+b\imath& c+d\imath\\-c+d\imath& a-b\imath\end{array}\right)=a\eins_2+b\imath\tau_3+c\imath\tau_2+d\imath\tau_1.
\end{equation}
The set of rectangular matrices is
\begin{eqnarray}\label{recmatdef}
 \Gl^{(\beta)}(n,m)&:=&\left\{\begin{array}{cl} \mathbb{R}^{n\times (n-m)}, & \beta=1, \\ \mathbb{C}^{n\times (n-m)}, & \beta=2, \\ \mathbb{H}^{n\times (n-m)}\subset\mathbb{C}^{2n\times 2(n-m)}, & \beta=4.\end{array}\right.
\end{eqnarray}
Alternatively, for $\beta=4$, we can introduce the rectangular matrices as
\begin{equation}
\label{descimH}
 \Gl^{(4)}(n,m):=\{A\in \mC^{2n\times (2n-2m)}\,|\,A^*=\tau_2^{(n)}A\tau_2^{(n-m)} \},
\end{equation}
where $A^\ast$ is the complex conjugation of the matrix $A$ and $\tau_2^{(k)}=\left[\begin{array}{cc} 0 & -\imath\eins_k \\ \imath\eins_k & 0 \end{array}\right]$ the constant enlarged second Pauli matrix.

We denote by $(.)^\dagger$ is the adjoint conjugation (transposition and complex conjugation) on $\Gl^{(\beta)}(n,m)$.
Then we have the manifold morphism\footnote{Note that for $\beta=1$ and $m=0$ this would give $O(n)$ rather than $SO(n)$. We ignore this as we are interested in non-trivial Stiefel manifolds, so $m>0$ where we have $O(n)/O(m)\cong SO(n)/SO(m)$.}
\begin{equation}
\label{represStiefel}
\St^{(\beta)}(n,m)\cong\{A\in \Gl^{(\beta)}(n,m)\,|\, A^\dagger A=\eins_{\gamma(n-m)} \}.
\end{equation}
Note that for any $A\in\Gl^{(\beta)}(n,m)$ the dyadic matrix $A^\dagger A$ is real symmetric ($\beta=1$), Hermitian ($\beta=2$), and Hermitian self-dual ($\beta=4$).

Here we stress that $\St^{(\beta)}(n,m)\hookrightarrow \Gl^{(\beta)}(n,m)$ is an embedding of real vector spaces or real manifolds and we will always interpret $\Gl^{(\beta)}(n,m)$ as a real vector space, despite the possible interpretation for $\beta=2,4$ as a complex vector space.

Consider an undetermined matrix $B\in\Gl^{(\beta)}(n,m)$. This matrix can be written as $B=U_L \Lambda U_R$ with $U_L\in \rmU^{(\beta)}(n)$, $U_R\in \rmU^{(\beta)}(n-m)$ and $\Lambda$ a (rectangular) diagonal matrix in $\mR^{\gamma n\times \gamma (n-m)}$. The $n-m$ diagonal elements of $\Lambda$ (for $\beta=4$ they are $2(n-m)$ but doubly degenerate) are referred to as the singular values of $B$. Correspondingly, the $\rmU^{(\beta)}(n)\times \rmU^{(\beta)}(n-m)$-invariant functionals on $\Gl^{(\beta)}(n,m)$ are an algebra generated by $n-m$ polynomial functions on $\Gl^{(\beta)}(n,m)$. We denote a basis of this algebra by $I_1,\cdots, I_{n-m}$. The basis can be given by $B\mapsto \tr (B^\dagger B)^k$ for $k=1,\cdots,n-m$. For $\beta\in\{1,2\}$ we choose the alternative basis $\{I_j,j=1,\cdots,n-m\}$, defined by
\begin{equation}\label{invariantsI}
\det(\eins_{n-m}+\mu B^\dagger B)=1+\sum_{j=1}^{n-m}\mu^jI_j(B),
\end{equation}
for an indeterminate $\mu$. In particular we obtain $I_1(B)=\tr(B^\dagger B)$, $I_2(B)=[\left(\tr B^\dagger B\right)^2-\tr\left( B^\dagger B\right)^2]/2$ and $I_{n-m}(B)=\det(B^\dagger B)$.\footnote{The case $\beta=1$, $m=0$ is again a bit more subtle as there $\det (B)$, the square root of $I_n(B)$, is also an invariant polynomial. Also here the results for $\beta=1$, $m=0$ correspond to $O(n)$.} Furthermore we can write higher powers of $B^\dagger B$ in terms of the first $n-m$ and the invariants $I_j(B)$ for $j\in[1,n]$ by using the Cayley-Hamilton theorem for an arbitrary $C\in\mC^{k\times k}$,
\begin{eqnarray}\label{recursion}
 \sum_{j=0}^{k}\left(\int_0^{2\pi}d\varphi e^{-\imath j\varphi}\det(e^{\imath\varphi}\eins_{n-m}-C)\right)C^j=0.
\end{eqnarray}
For $\beta=4$ we define $I_j(B)$ by the expansion
\begin{equation}\label{invariantsI-beta4}
{\rm Pf}(I-\mu IB^\dagger B)=1+\sum_{j=1}^{n-m}\mu^jI_j(B)
\end{equation}
with the constant $2(n-m)\times 2(n-m)$ matrix $I$ defined as the block diagonal matrix
\begin{eqnarray}\label{I-def}
 I=\imath\tau_{2}^{(n-m)}.
\end{eqnarray}
The Pfaffian determinant in Eq.~\eqref{invariantsI-beta4}, see \cite{Mehtabook}, is only defined for even dimensional, anti-symmetric matrices $C$ as
\begin{eqnarray}\label{Pfaffiandef}
 \Pf\,C=\frac{1}{N!}\sum\limits_{\omega\in\mathcal{S}_{2N}}\sign\omega \prod\limits_{j=1}^N\frac{C_{\omega(2j-1)\omega(2j)}}{2}.
\end{eqnarray}
A Pfaffian determinant is essentially an exact square root of a determinant yielding again a polynomial in the matrix elements of $C$. We employed the sign function ``$\sign$" which is $+1$ for an even permutation $\omega\in\mathcal{S}_{2N}$ of $2N$ elements and $-1$ for an odd one.

We denote real symmetric, Hermitian and Hermitian self-dual matrices by
\begin{eqnarray}\label{Hermdef}
 \Herm^{(\beta)}(n-m):=\left\{\begin{array}{cl} \rmUn(n-m)/\so(n-m), & \beta=1, \\ \gl(n-m)/\rmUn(n-m), & \beta=2, \\ \rmUn(2(n-m))/\rmUsp(2(n-m)), & \beta=4,\end{array}\right.
\end{eqnarray}
with the real Lie-algebras  $\rmUn(n-m)$, $\so(n-m)$, $\gl(n-m)$, and $\rmUsp(2(n-m))$ of the groups $\rmU(n-m)$, $\SO(n-m)$, $\Gl(n-m)$, and $\rmUSp(2(n-m))$, respectively.

In Pizzetti's formula the renormalized Bessel function plays an important role. It is defined as
\begin{equation}\label{renormBessel}
\Psi_{N/2-1}(x):=\frac{\Gamma(N/2)J_{N/2-1}(x)}{\left(x/2\right)^{N/2-1}},\ \Psi_{N/2-1}(0)=1,
\end{equation}
where $J_\nu$ is the Bessel function of the first kind and of order $\nu$. Moreover we need in our calculations  the floor function $\lfloor x\rfloor$ to denote the largest integer smaller than or equal to $x\in\mathbb{R}$.

Finally let us consider two vector variables $u,v\in \mR^d$ for some $d\in\mN$. Then we introduce the following differential operators on $\mR^{2d}$:
\begin{equation}\label{harmdiff1}\Delta_u=\sum_{j=1}^n\partial_{u_j}^2,\,\,\,\langle \nabla_u,\nabla_v\rangle=\sum_{j=1}^n\partial_{u_j}\partial_{v_j},\,\,\, \langle u,\nabla_v\rangle=\sum_{j=1}^nu_j\partial_{v_j},\end{equation}\begin{equation}\label{harmdiff2} \mE_u=\sum_{j=1}^nu_j\partial_{u_j},\,\,\, \langle u, v\rangle=\sum_{j=1}^nu_jv_j\mbox{ and }\,\,u^2=\sum_{j=1}^nu_j^2.\end{equation}
These operators are important for deriving the Pizzetti formula of the Stiefel manifolds $\St^{(\beta)}(n,n-2)$. The following commutation relations of these operators are immediate.

\begin{lemma}
\label{commrel}
The operators in equations \eqref{harmdiff1} and \eqref{harmdiff2} satisfy the following commutation relations:
\begin{eqnarray}
[\Delta_u,u^2]=4\mE_u+2n,&&[\langle u,\nabla_v\rangle,\langle \nabla_u,\nabla_v\rangle]=-\Delta_v,\nonumber\\
\,[\langle\nabla_u,\nabla_v\rangle,u^2]=2\langle u,\nabla_v\rangle, && [\Delta_u,\langle u,\nabla_v\rangle]=2\langle \nabla_u,\nabla_v\rangle,\nonumber\\
\,[\langle v, \nabla_u\rangle,u^2]=2\langle u,v\rangle, && [\langle u,\nabla_v\rangle,\langle u,v\rangle]=u^2,\nonumber\\ 
\,[\Delta_u,\langle u,v\rangle]=2\langle v,\nabla_u\rangle, && [\langle\nabla_u,\nabla_v\rangle,\langle u,v\rangle]=\mE_u+\mE_v+m.
\end{eqnarray}
\end{lemma}


\section{Motivation and sketch of the main idea}\label{sketch}

In this section we want to sketch the main idea behind a Pizzetti formula for integrals over some particular coset integrals, namely the Stiefel manifolds \eqref{Stiefel}. These comprise two special cases, namely the full group $\St^{(\beta)}(n,0)=\rmU^{(\beta)}(n)$ and the coset $\St^{(\beta)}(n,n-1)\cong\mS^{\beta n-1}$ which is isomorph to  the $(\beta n-1)$-dimensional unit sphere.

Please note that all calculations done in this section are sketches only, and shall only serve for a better understanding of the main results of this work and their importance in calculations of problems in random matrix theory. Thus we do not work out all technical requirements in detail in this section. In particular we assume that the whole calculation  smoothly works out although there are some intermediate steps which can be criticized. Nevertheless most of the results, derived in a hand waving way in this section, are justified since they are rigorously proven in the ensuing sections. Hence they serve as educated guesses and good starting points for mathematical proofs.

\subsection{The main idea of the classical Pizzetti formula}\label{originalPiz}

Let $N>1$ be an integer, the Pizzetti formula, see \cite{MR2885533, Pizzetti}, expresses integration over the unit sphere of a function, of which the Taylor series at the origin converges in a neighbourhood containing the unit sphere, in terms of differential operators as
\begin{equation}\label{classPizzetti}
 \int_{\mS^{N-1}}f=\sum_{j=0}^\infty\frac{\Gamma(N/2)}{4^j j!\Gamma(j+N/2)}(\Delta^j f)(0)=\left(\Psi_{N/2-1}(\sqrt{-\Delta})f\right)(0),
\end{equation}
where $\Gamma$ is the Euler $\Gamma$-function, $\Delta$ is the flat Laplacian on the $N$-dimensional space $\mathbb{R}^N$ and $\Psi_{N/2-1}$ as in \eqref{renormBessel}. This formula was originally introduced by Pizzetti as a generalization of the mean value theorem for harmonic functions to poly-harmonic functions.

In this subsection we derive and interpret this formula in an alternative fashion. As an extra result of this we find that formula \eqref{classPizzetti} is also applicable to arbitrary functions in the Schwartz space of rapidly decreasing functions $\cS(\mR^N)$.

Consider a function $f$ on the unit sphere. To render the identity~\eqref{classPizzetti} meaningful we have to assume that the sphere $\mS^{N-1}$ is embedded in $\mathbb{R}^N$ and that there exists a smooth extension of $f$ to an open set covering the $N$-dimensional ball bounded by $\mS^{N-1}$. From a physics point of view Eq.~\eqref{classPizzetti} is quite natural since it shows the relation between the configuration space and the momentum space in the quantum mechanical framework. We make this more concrete and consider the integral
\begin{equation}\label{classPizzettider1}
 \int_{\mS^{N-1}}f= \frac{\int_{\mathbb{R}^N}d[x]f(x)\delta(x^2-1)}{\int_{\mathbb{R}^N}d[x]\delta(x^2-1)}.
\end{equation}
Here we assume that there is a smooth extension of $f$ into $\mathbb{R}^N$. The measure $d[x]$ is the flat measure on $\mathbb{R}^N$ and is, thus, the product of the differentials of the coordinates of the $N$-dimensional vector $x$. The Dirac $\delta$-distribution restricts the integral onto the sphere $\mS^{N-1}$ and the integral in the denominator normalizes the integral such that $\int_{\mS^{N-1}}1=1$. Assuming that the Fourier transform of $f$ exists, i.e.
\begin{equation}\label{Fourierf}
 \mathcal{F}[f](k):=\frac{1}{(2\pi)^{N/2}}\int_{\mathbb{R}^N}d[x] \exp[-\imath \langle k,x\rangle] f(x),
\end{equation}
 we can rewrite the integral~\eqref{classPizzettider1},
\begin{equation}\label{classPizzettider2}
 \int_{\mS^{N-1}}f= \frac{\int_{\mathbb{R}^N}d[x]\delta(x^Tx-1)\int_{\mathbb{R}^N}d[k] \exp[\imath \langle k,x\rangle] \mathcal{F}[f](k)}{(2\pi)^{N/2}\int_{\mathbb{R}^N}d[x]\delta(x^2-1)}.
\end{equation}
Here $\langle.,.\rangle$ is the standard Euclidean scalar product on $\mathbb{R}^N$.
Interchanging the integrals\footnote{Again we emphasize that this can be made rigorous but technical details obscure the idea behind this approach.} over $x$ and $k$, the integral over $x$ is an average over the azimuthal angle ${\rm arccos}(\langle k,x\rangle/\sqrt{x^2k^2})$ of a plane wave which is equal to the renormalized Bessel function~\eqref{renormBessel}, i.e.
\begin{equation}\label{angleaverage}
\frac{\int_{\mathbb{R}^N}d[x]\delta(x^2-1)\exp[\imath \langle k,x\rangle]}{\int_{\mathbb{R}^N}d[x]\delta(x^2-1)}=\Psi_{N/2-1}(\sqrt{k^2}) .
\end{equation}
Hence the integral~\eqref{classPizzettider1} is
\begin{equation}\label{classPizzettider3}
 \int_{\mS^{N-1}}f= \frac{\int_{\mathbb{R}^N}d[k] \Psi_{N/2-1}(\sqrt{k^2}) \mathcal{F}[f](k)}{(2\pi)^{N/2}}.
\end{equation}
When we plug the definition of the Fourier transform~\eqref{Fourierf} into this integral we have to replace the wave vector $k$ by the gradient $\imath\nabla$ in the spatial vector $x$ and, thus, the norm $k^2$ by minus the Laplacian $(\imath\nabla)^2=-\Delta$. The remaining integral over $k$ is equal to an $N$-dimensional Dirac $\delta$-distribution telling us that we have to evaluate the spatial vector $x$ at the origin. This yields Pizzetti's formula~\eqref{classPizzetti}.

Summarizing this calculation, this integral is equal to an action of an operator $\Psi_{N/2-1}(\sqrt{-\Delta})$ only involving momentum operators conjugate to the spatial vector $x$. The Fourier transformation connects both representation. The origin of the operator $\Psi_{N/2-1}(\sqrt{-\Delta})$ lies in the connection through the Fourier transform with the Dirac delta function characterizing the unit sphere in Eq. \eqref{angleaverage}.

Note the difference of Eq.~\eqref{classPizzetti} to an integration over the full space $\mathbb{R}^N$.  Then the integral would be equal to  $\mathcal{F}[f](0)$ meaning that we have to evaluate the Fourier transform at the origin and not the function $f$ itself in contrast to Eq.~\eqref{classPizzetti}.

\subsection{A Pizzetti formula for groups and Stiefel manifolds}\label{newPiz}

We generalize the approach presented in subsection~\ref{originalPiz} to general Stiefel manifolds. According to Eq.~\eqref{represStiefel}, $\St^{(\beta)}(n,m)$ can be interpreted as a boundary of a domain in $\Gl^{(\beta)}(n,m)$ containing the origin. We consider a function $f$ on $\St^{(\beta)}(n,m)$ which smoothly extends to $\Gl^{(\beta)}(n,m)$. We argue that the Haar measure on the Stiefel manifold $\St^{(\beta)}(n,m)$ ($0\leq m< n$) can be represented as\footnote{Also here, for $\beta=1$ and $m=0$ we would get integration over $O(m)$ rather than $SO(m)$.}
\begin{eqnarray}\label{Stiefintdef}
 \int_{\St^{(\beta)}(n,m)}f=\frac{\int_{\Gl^{(\beta)}(n,m)}d[A]f(A)\delta(A^\dagger A-\eins_{\gamma (n-m)})}{\int_{\Gl^{(\beta)}(n,m)}d[A]\delta(A^\dagger A-\eins_{\gamma (n-m)})}.
\end{eqnarray}
The Dirac $\delta$-distribution, restricting the dyadic matrix $A^\dagger A$ to unity, is defined for real symmetric, Hermitian, and Hermitian self-dual matrices, respectively, by the product of the Dirac $\delta$-distributions of each real independent degree of freedom in the matrix. One can easily check that the Dirac $\delta$-distribution with the flat measure $d[A]$ builds the Haar measure of the coset $\St^{(\beta)}(n,m)$, i.e. the measure is invariant under $A\to U_{\rm L}AU_{\rm R}$ with any $U_{\rm L}\in\rmU^{(\beta)}(n)$ and  $U_{\rm R}\in\rmU^{(\beta)}(n-m)$, showing that Eq.~\eqref{Stiefintdef} is correct. The denominator again normalizes the integral, i.e. $\int_{\St^{(\beta)}(n,m)}1=1$.
 
We assume that the Fourier transform of $f$ as a function on $\Gl^{(\beta)}(n,m)$ exists, i.e.
\begin{eqnarray}\label{Fourierfmat}
 \mathcal{F}[f](B):=\frac{1}{(2\pi)^{\beta n(n-m)/2}}\int_{\Gl^{(\beta)}(n,m)}d[A] \exp\left[-\frac{\imath}{2\gamma}(\tr B^\dagger A+\tr A^\dagger B)\right] f(A),
\end{eqnarray}
for $B\in \Gl^{(\beta)}(n,m)$. Mimicking step~\eqref{angleaverage} we need to consider the integral \eqref{Stiefintdef} with $f$ replaced by the Fourier kernel in Eq.~\eqref{Fourierfmat}. This yields a Leutwyler-Smilga-like integral \cite{LeutwylerSmilga} and we denote the result by $\widehat{\Psi}^{(\beta)}_{n,m}(B)$, see Eq. \eqref{groupint}. This integral plays an important role in chiral perturbation theory of QCD \cite{ShuryakVerbaarschot,Verbaarschot,LeutwylerSmilga,SchlittgenWettig,KVZ14} and in random matrix theory \cite{Zirnbauer96}.

Note that $ \widehat{\Psi}^{(\beta)}_{n,m}(U_LBU_R) =\widehat{\Psi}^{(\beta)}_{n,m}(B)$ for $U_L\in \rmU^{(\beta)}(n)$ and $U_R\in \rmU^{(\beta)}(n-m)$. Hence $ \widehat{\Psi}^{(\beta)}_{n,m}(B)$ only depends on the singular values of $B$ or alternatively on the matrix invariants, see section \ref{secprel} for introduction of both concepts. This justifies the other two functions introduced in the definition underneath.

\begin{definition}
We define the function $\widehat{\Psi}^{(\beta)}_{n,m}$ on $\Gl^{(\beta)}(n,m)$ by
\begin{eqnarray}\label{groupint}
 \widehat{\Psi}^{(\beta)}_{n,m}(B):=\frac{\int_{\Gl^{(\beta)}(n,m)}d[A]\delta(A^\dagger A-\eins_{\gamma (n-m)})\exp\left[\imath(\tr B^\dagger A+\tr A^\dagger B)/(2\gamma)\right]}{\int_{\Gl^{(\beta)}(n,m)}d[A]\delta(A^\dagger A-\eins_{\gamma (n-m)})}.
\end{eqnarray}
and the function $\widetilde{\Psi}^{(\beta)}_{n,m}$ on $\mR^{n-m}$ by
\begin{equation} \label{phitilde} \widetilde{\Psi}^{(\beta)}_{n,m}(\Lambda_1,\cdots,\Lambda_{n-m}):= \widehat{\Psi}^{(\beta)}_{n,m}(B),\end{equation}
for any $B\in \Gl^{(\beta)}(n,m)$ with the set of singular values $\{\Lambda_1,\cdots,\Lambda_{n-m}\}$. Furthermore we define an element $\Psi_{n,m}^{(\beta)}(x_1,\cdots,x_{n-m})$ in the space of formal power series $\mR[[x_1,\cdots,x_{n-m}]]$ in $n-m$ variables and with real coefficients, by setting
\begin{equation}\label{defphi}
\Psi_{n,m}^{(\beta)}(I_1(B),\cdots,I_{n-m}(B)):=\widehat{\Psi}^{(\beta)}_{n,m}(B),\quad\mbox{
for any $B\in \Gl^{(\beta)}(n,m)$}.
\end{equation}
\end{definition}
The case $n-m=1$ yields the original Pizzetti formula for all three Dyson indices $\beta=1,2,4$. This follows from section \ref{sec2}, but also directly from the fact that the unique $\rmU^{(\beta)}(n)$-invariant integration on $\St^{(\beta)}(n,n-1)\cong \mS^{\beta n-1}$ must correspond to the unique $\SO(\beta n)$-invariant integration because of the embedding $\rmU^{(\beta)}(n)\hookrightarrow \SO(\beta n)$. This implies
\begin{equation}
\label{connPizz}
\widehat{\Psi}^{(\beta)}_{n,n-1}(B)=\Psi_{\beta n/2-1}(\Lambda_1)=\Psi_{\beta n/2-1}(\sqrt{I_1(B)}),
\end{equation}
so also
\begin{equation}
\label{connPizz2}
\widetilde{\Psi}^{(\beta)}_{n,n-1}=\Psi_{\beta n/2-1}\quad\mbox{and}\quad \Psi_{n,n-1}^{(\beta)}=\Psi_{\beta n/2-1}\circ \sqrt{\hspace{2mm}}.
\end{equation}

First we consider the unitary group ($\beta=2$). For this case, integral~\eqref{groupint} is known~\cite{LeutwylerSmilga,SchlittgenWettig} and yields
\begin{eqnarray}\label{groupintunitary}
 \widehat{\Psi}^{(2)}_{n,m}(B)=\frac{1}{\Delta_{n-m}(\Lambda^2)}\det\left[\Lambda_a^{2(n-m-b)}\Psi_{n-b}(\Lambda_a)\right]_{1\leq a,b\leq n-m},
\end{eqnarray}
where we employed the singular value decomposition of $B^\dagger B=V\Lambda^2 V^\dagger$ with $\Lambda$ positive definite and diagonal and $V\in\rmU^{(\beta)}(n-m)$. In the denominator we used the Vandermonde determinant
\begin{eqnarray}\label{Vanddef}
\Delta_{n-m}(\Lambda^2)=\prod\limits_{1\leq a<b\leq n-m}(\Lambda^2_a-\Lambda^2_b)=\det\left[\Lambda_a^{2(n-m-b)}\right]_{1\leq a,b\leq n-m}.
\end{eqnarray}
In the case $m=n-1$ the Vandermonde determinant is trivial and the integral becomes the classical Pizzetti formula~\eqref{classPizzetti} for an $\mS^{2n-1}$-sphere, see \eqref{connPizz}. 

To arrive at a Pizzetti-like formula we have to express the function $\widehat{\Psi}^{(2)}$ in terms of the full matrix $B$ instead of its singular values $\Lambda$, only. Therefore, both in denominator and numerator, we multiply with the Vandermonde determinant $\Delta_{n-m}(\Lambda^2)$ and use the identity $\det M_1\det M_2= \det M_1M_2$ for $M_1,M_2\in\mC^{d\times d}$,
\begin{eqnarray}\label{groupintunitary.b}
 \widehat{\Psi}^{(2)}_{n,m}(B)&=&\frac{\det\left[\tr\left(\Lambda^{2(a+b-2)}\Psi_{m+b-1}(\Lambda)\right)\right]_{1\leq a,b\leq n-m}}{\det\left[\tr\Lambda^{2(a+b-2)}\right]_{1\leq a,b\leq n-m}}\\
 &=&\frac{\det\left[\tr\left((B^\dagger B)^{(a+b-2)}\Psi_{m+b-1}(\sqrt{B^\dagger B})\right)\right]_{1\leq a,b\leq n-m}}{\det\left[\tr(B^\dagger B)^{(a+b-2)}\right]_{1\leq a,b\leq n-m}}.\nonumber
\end{eqnarray}
Now we replace the matrix $B$ by the gradient which is an  $ n \times(n-m)$ matrix of partial derivatives,
\begin{eqnarray}\label{gradient-def}
 \{\nabla\}_{ji}=\frac{\beta}{\gamma}\frac{\partial}{\partial A_{ji}};\ 1\leq j\leq n, 1\leq i\leq n-m.
\end{eqnarray}
Then we arrive at
\begin{eqnarray}\label{Pizzetti-beta2}
\int_{\St^{(2)}(n,m)}f
 &=&\left(\frac{\det\left[\tr\left((-\nabla^\dagger \nabla)^{(a+b-2)}\Psi_{m+b-1}(\sqrt{-\nabla^\dagger \nabla})\right)\right]_{1\leq a,b\leq n-m}}{\det\left[\tr(-\nabla^\dagger \nabla)^{(a+b-2)}\right]_{1\leq a,b\leq n-m}}f\right)(0).
\end{eqnarray}
This result is quite compact but one has to be careful since the operator in the denominator may contain zero eigenvalues which cancel with those in the numerator. Hence one has to expand the function in the gradient and it can be checked that this yields a Taylor series in $\nabla^\dagger\nabla$. Expression \eqref{Pizzetti-beta2} is written in terms of the operators $\tr (\nabla^\dagger\nabla)^j$ for $j\in\mathbb{N}$. Recall that all these traces can be expressed in the coefficients of the characteristic polynomial $\det(\nabla^\dagger\nabla-\varepsilon\eins_{n-m})$, introduced as $\{I_k(\nabla^\dagger\nabla),j\in [1,n-m]\}$ in section \ref{secprel}, due to the relation \eqref{recursion}. Nevertheless Eq.~\eqref{Pizzetti-beta2} remains a conjecture since we have not found any clean proof, yet. For the case $m=n-2$ (which will be proved rigorously) and the results in section \ref{PizzInv} we will work with formulae that are expressed explicitly in terms of the generators $\{I_k(\nabla^\dagger\nabla),j\in [1,n-m]\}$.

For the quaternion case ($\beta=4$) we are able to obtain the analogue of the result~\eqref{Pizzetti-beta2}, since the coset integral~\eqref{groupint} can also be performed, for arbitrary $m$ and $n$, exactly. The reason is that this integral fulfills a particular underlying algebraic structure, namely a Pfaffian determinant~\eqref{Pfaffiandef}. First we define two families of matrix valued functions $G^{(m)}_{ab}:\mC^{d\times d}\to \mC^{d\times d}\otimes \mC^{d\times d}$ and $g^{(m)}_{ab}:\mC^{d\times d}\to \mC^{d\times d}$, for arbitrary $d\in\mN$ and $1\le a,b\le n-m$ as
\begin{eqnarray}
 G_{ab}^{(m)}(X)&=&(X^{a}\otimes X^{b-1}-X^{a-1}\otimes X^{b})\widetilde{\Psi}^{(4)}_{m+2,m}(\sqrt{X}\otimes\eins_{2(n-m)},\eins_{2(n-m)}\otimes \sqrt{X}),\nonumber\\
 g_{a}^{(m)}(X)&=&(X)^{a-1}\Psi_{2m+1}(\sqrt{X}).\label{aux-func}
\end{eqnarray}
As derived in Appendix~\ref{app1}, the Pizzetti formulae for Stiefel manifolds $\St^{(4)}(n,m)$ read, {\textbf{for $n-m$ even}}
\begin{eqnarray}\label{Pizzetti-beta4-ev}
\int_{\St^{(4)}(n,m)}f
 &=&\left(\prod\limits_{j=1}^{n-m}\frac{2^{2j-5}\Gamma(2j+2m-1)}{\sqrt{\Gamma(2m+3)\Gamma(2m+1)}}\right)\left(\frac{\Pf\left[\tr G_{ab}^{(m)}(-\nabla^\dagger \nabla)\right]_{1\leq a,b\leq n-m}}{\det\left[\tr(-\nabla^\dagger \nabla)^{(a+b-2)}\right]_{1\leq a,b\leq n-m}}f\right)(0),
\end{eqnarray}
and \textbf{ for $n-m$ odd} 
\begin{eqnarray}\label{Pizzetti-beta4-odd}
\int_{\St^{(4)}(n,m)}f
 &=&\left(\prod\limits_{j=1}^{n-m-1}\frac{2^{2j+1}\Gamma(2j+2m+1)}{\sqrt{\Gamma(2m+3)\Gamma(2m+1)}}\right)\\
 &&\times\left(\frac{\Pf\left[\begin{array}{cc} 0 & \left\{\tr g_{b}^{(m)}(-\nabla^\dagger \nabla)\right\}_{1\leq b\leq n-m} \\ \left\{-\tr g_{a}^{(m)}(-\nabla^\dagger \nabla)\right\}_{1\leq a\leq n-m} & \left\{\tr G_{ab}^{(m)}(-\nabla^\dagger \nabla)\right\}_{1\leq a,b\leq n-m}  \end{array}\right]}{\det\left[\tr(-\nabla^\dagger \nabla)^{(a+b-2)}\right]_{1\leq a,b\leq n-m}}f\right)(0).\nonumber
\end{eqnarray}
Again this result exhibits intriguing structures well-known in random matrix theory but we have no rigorous proof for it such that it remains a conjecture.

Note that these formulae still employ the unknown function $\widetilde{\Psi}^{(4)}_{m+2,m}$. In Eq. \eqref{def-psi-tilde.b} in appendix \ref{app1} we express this function in terms of the Bessel function $\Psi_m$. Alternatively one can use Eq. \eqref{groupint-bet4.d} in appendix~\ref{app0}.

For the real case ($\beta=1$) the situation is much more involved and the integral~\eqref{groupint} is only known for particular $B$, see \cite{Zirnbauer96,KVZ14} in combination with \cite{AltlandZirnbauer97,Ivanov02,Mehtabook}. In the case when the spectrum of $B$ is doubly degenerate the integral~\eqref{groupint} is equivalent to the microscopic limit (see Refs.~\cite{VerbaarschotWettig,Gernotbook} for the meaning of the notion) of a partition function of specific random matrix ensembles. This random matrix ensemble can then be solved exactly with the help of orthogonal polynomials (see Refs.~\cite{Mehtabook,Peterbook,Gernotbook} for the method of orthogonal polynomials). 

Though the Pizzetti formulae~\eqref{Pizzetti-beta2}, \eqref{Pizzetti-beta4-ev}, and \eqref{Pizzetti-beta4-odd} are quite compact, they still seem to be cumbersome to evaluate in realistic situations. However for the cases $n-m=1$ and $n-m=2$ the formulae simplify a lot. For the case $n-m=2$ we pursue another approach, presented in appendix~\ref{app0}, which yields
\begin{eqnarray}\nonumber
\int_{\St^{(\beta)}(n,n-2)}f&=&\sum_{j=0}^\infty\frac{\Gamma(\beta n/2)\Gamma(\beta(n-1)/2)}{4^j\Gamma(\beta n/2+j)\Gamma(\beta (n/2-1)+j)}\left(\left.\frac{1}{j!}\frac{\partial^j}{\partial\mu^j}\right|_{\mu=0}{\det}^{\beta(n-1)/2+j-1}(\mu \nabla^\dagger\nabla+\eins_2)f\right)(0)\\
 &=&\sum_{j=0}^\infty\sum_{l=0}^{\lfloor j/2\rfloor}\frac{\Gamma(\beta n/2)\Gamma(\beta(n-1)/2)}{4^j\Gamma(\beta n/2+j)\Gamma(\beta(n-1)/2+l)(j-2l)!l!}\left(\Delta^{j-2l}{\det}^{l}(\nabla^\dagger\nabla) f\right)(0)\label{Pizzetti-beta2-2bb}
 \end{eqnarray}
for $\beta=1,2$ with $\Delta=\tr \nabla^\dagger\nabla$. For $\beta=4$ we get similarly
\begin{eqnarray}
\int_{\St^{(4)}(n,n-2)}f&=&\sum_{j=0}^\infty\frac{\Gamma(2n)\Gamma(2(n-1))}{4^j\Gamma(2n+j)\Gamma(2(n-1)+j)}\left(\left.\frac{1}{j!}\frac{\partial^j}{\partial\mu^j}\right|_{\mu=0}{\Pf}^{2n-3+j}(\mu I\nabla^\dagger\nabla+I)f\right)(0)\nonumber\\
&=&\sum_{j=0}^\infty\sum_{l=0}^{\lfloor j/2\rfloor}\frac{\Gamma(2n)\Gamma(2(n-1))}{4^j\Gamma(2n+j)\Gamma(2(n-1)+l)(j-2l)!l!}\left(\frac{\Delta}{2}\right)^{j-2l}{\Pf}^{l}(I\nabla^\dagger\nabla)\label{Pizzetti-beta4-2b}.
\end{eqnarray}
These results are rigorously proven in section \ref{sec2}. The multiplication with the $4\times 4$-matrix $I$ as in Eq.~\eqref{I-def}, is necessary because it transforms a Hermitian self-dual matrix into an anti-symmetric matrix such that the Pfaffian of those matrices is well-defined. The factor $1/2$ in the term $\left(\Delta/2\right)^{j-2l}$ normalizes the Laplacian correctly because it is $\Delta=\tr(\nabla^\dagger\nabla)$ and the matrix $\nabla^\dagger\nabla$ is Kramers degenerate.

\subsection{Stiefel manifolds in non-linear $\sigma$-models and random matrix theory}\label{appl-mot}

The Leutwyler-Smilga-like integral~\eqref{groupint} is one of the integrals which quite naturally appear in non-linear $\sigma$ models and in random matrix theories \cite{ShuryakVerbaarschot,Verbaarschot,Zirnbauer96,VerbaarschotWettig,KVZ14}.
For instance, in the microscopic limit of four-dimensional continuum QCD, the integral
\begin{eqnarray}\label{LS-bet2}
 Z(M)=\int_{\rmU(N_{\rm f})}d\mu(U){\det}^\nu U\exp\left[\frac{1}{2}\tr (UM +U^\dagger M^\dagger)\right],
\end{eqnarray}
represents the partition function of $N_{\rm f}$ flavors of quarks in a ${\rm SU}(N_{\rm c}>3)$ gauge theory where the quarks are in the fundamental representation \cite{ShuryakVerbaarschot}. In this subsection $d\mu$ always stands for the unique invariant measure under the group action on a certain coset of the group. The determinant in $U$, to the power $\nu\in\mathbb{Z}$, reflects the non-trivial topological configurations in this gauge theory. The mass matrix $M$ is sometimes chosen as an arbitrary complex matrix to generate particular observables. Quite often the matrix $M$ is chosen diagonal and some of the singular values (usually two, for the up and the down quark) are degenerate, say $M={\rm diag}(m,m,m_1,\ldots,m_{N_{\rm f}-2})$. Then the integral factorizes in an integral over the invariant group $\rmU(2)$ and over the coset $\rmU(N_f)/\rmU(2)$. The integral over $\rmU(2)$ yields an overall constants while the remaining integral runs over a Stiefel manifold.

Also the counterparts of Eq.~\eqref{LS-bet2} over the groups ${\rm O}(2N_{\rm f})$,
\begin{eqnarray}\label{LS-bet1}
 Z(M)=\int_{{\rm O}(2N_{\rm f})}d\mu(U){\det}^\nu U\exp\left[\tr UM\right],\ \nu\in\{0,1\},
\end{eqnarray}
and over $\rmUSp(2N_{\rm f})$,
\begin{eqnarray}\label{LS-bet4}
 Z(M)=\int_{\rmUSp(2N_{\rm f})}d\mu(U)\exp\left[\tr UM\right],
\end{eqnarray}
are found in QCD-like theories, namely two-dimensional QCD for a ${\rm SU}(N_{\rm c}\geq2)$ gauge theory with the fermions in the adjoint representation and for the ${\rm SU}(N_{\rm c}=2)$ gauge theory with the fermions in the fundamental representation, respectively, see Ref.~\cite{KVZ14}.

The integrals~\eqref{LS-bet2}, \eqref{LS-bet1}, and \eqref{LS-bet4} only become integrals over Stiefel manifolds when choosing particular masses. When considering the local statistics of Hermitian random matrices in the limit of large matrix size, the Stiefel manifolds always arise naturally from the saddlepoint approximation. The corresponding non-linear $\sigma$-model is
\begin{eqnarray}\label{part-loc}
 Z(M)&=&\int_{\rmU^{(\beta)}(2N_{\rm f})/[\rmU^{(\beta)}(N_{\rm f})\times\rmU^{(\beta)}(N_{\rm f})]}d\mu(U)\exp\left[\tr U{\rm diag}(\eins_{\gamma N_{\rm f}},-\eins_{\gamma N_{\rm f}})U^\dagger M\right]\\
&=&\exp[\tr M]\int_{\St^{(\beta)}(2N_{\rm f},N_{\rm f})}d\mu(A)\exp\left[-2\tr AA^\dagger M\right],\nonumber
\end{eqnarray}
for $\beta\in\{1,2,4\}$, see \cite{Zirnbauer96}. In the second line we rewrote integral~\eqref{part-loc} as an integral over the Stiefel manifold $\St^{(\beta)}(2N_{\rm f},N_{\rm f})$, which is immediate since it only differs by a global constant. Integral \eqref{part-loc} is a particular case of the Itzykson-Zuber integral~\cite{ItzZub} which should not be confused with the Harish-Chandra integral~\cite{Harish-Chandra}. The integral~\eqref{part-loc} is the fundamental connection between a large number of completely different theories like random matrix theory \cite{Gernotbook}, QCD-like theories in odd-dimensions \cite{VerbaarschotZahed}, and non-linear $\sigma$-models \cite{Zirnbauer96}. 

 Another important case of the Itzykson-Zuber integral~\cite{ItzZub} appearing quite often in random matrix theory is the integral
\begin{eqnarray}\label{ItzZub}
 Z(H)&=&\int_{\SO(4)/[\SO(2)\times\SO(2)]}d\mu(U)\exp\left[\tr U{\rm diag}(\eins_{2},-\eins_{2})U^T H\right]\\
&=&\exp[\tr H]\int_{\St^{(1)}(4,2)}d\mu(A)\exp\left[-2\tr AA^T H\right],\nonumber
\end{eqnarray}
with $H$ an arbitrary real diagonal matrix. This integral is one half of a supersymmetric integral to calculate the two-point correlation function of a real symmetric matrix, see Refs.~\cite{Heiner-thesis,GuhrKohler-b}. The hope is by finding a suitable expression of the integral~\eqref{ItzZub} one can derive analytical results for the transition between the Gaussian orthogonal ensemble (GOE) and the Poisson ensemble which naturally occurs in many physical systems~\cite{GM-GW,Haakebook}. We apply our results on integral~\eqref{ItzZub} in section~\ref{appl}.

\section{Pizzetti formulae from the point of view of invariant theory}\label{sec;Piz-beta1}
\label{secinv}

In this section we consider the idea of Pizzetti-type formulae for Stiefel manifolds from the point of view of classical invariant theory. This provides an alternative approach to the one in subsection~\ref{newPiz} and will lead (in section \ref{sec2}) to rigorous proofs of the results for $m=n-2$. We also introduce some Howe dual pairs and show the application of our results in that theory.

We denote the complexification of the Lie algebra of the groups $\rmU^{(\beta)}(n)$ by $\Li$, so concretely
\begin{equation}
\Li := \left\{\begin{array}{cl} \mathfrak{so}(n;\mC), & \beta=1, \\ \mathfrak{gl}(n;\mC), & \beta=2, \\ \mathfrak{sp}(2n;\mC), & \beta=4.\end{array}\right.
\label{algebras}
\end{equation}
Throughout this section we will denote the indeterminate matrix in $\Gl^{(\beta)}(n,m)$ by $X$ and by $\nabla_X$ the corresponding matrix of partial derivatives.

\subsection{Pizzetti formulae}

We have the inclusion $\St^{(\beta)}(n,m)\hookrightarrow \Gl^{(\beta)}(n,m)$ in equation \eqref{represStiefel}. In this subsection we consider invariant integration over $\St^{(\beta)}(n,m)$ of a broad class of functions on $\Gl^{(\beta)}(n,m)$ containing the algebra of polynomials. This space of functions is given by the following.

\begin{definition}
We define $\FS$ as the space of functions $f$ which belong to $\cC^\infty(\Omega)$ for some open set $\Omega\subset \Gl^{(\beta)}(n,m)$, such that $\{\lambda A| A\in\St^{(\beta)}(n,m),\lambda\in[0,1]\}\subset \Omega$, and the Taylor series of $f$ at the origin converges uniformly in $\Omega$.
\end{definition}

We recall the formal power series in Eq.~\eqref{defphi} and evaluate it in the differential operators on $\Gl^{(\beta)}(n,m)$ with constant coefficients $I_j,j\in[1,n-m]$, which we define as $I_j:= I_j(\imath\nabla_X)$.

\begin{theorem}
\label{PizzInv}
The integration over $\St^{(\beta)}(n,m)$ of $f\in \FS$ with respect to the left $\rmU^{(\beta)}(n)$-invariant measure is
\begin{equation}
\int_{\St^{(\beta)}(n,m)}f=\left(\Psi^{(\beta)}_{n,m}(I_1,\cdots,I_{n-m})f\right)(0),
\end{equation}
where the limit on the right-hand side converges.
\end{theorem}

\begin{proof}
First we consider the case where $f$ is a polynomial function on $\Gl^{(\beta)}(n,m)\cong \mR^{\beta n(n-m)}$. Any functional $T$ acting on the space of polynomials denoted by $\mR\left[\mR^{\beta n(n-m)})\right]$  can be written as 
\begin{equation}
T(f)=\sum_{j=0}^\infty (T_j f)(0)
\end{equation}
where $T_j$ are differential operators with constant coefficients of degree $j$.

Now imposing the condition that this functional is invariant for left $\rmU^{(\beta)}(n)$-action implies that each $T_j$ should be $\rmU^{(\beta)}(n)$-invariant. 
If we furthermore impose that $T((X^\dagger X)_{ij}\cdot)\equiv \delta_{ij}T(\cdot)$ for $1\le i,j\le n-m$, meaning that we restrict to $\St^{(\beta)}(n,m)$, we know that $T$ is also invariant under the right $\rmU^{(\beta)}(n-m)$-action. All of this implies that each $T_j$ is in the algebra generated by $\{I_1,\cdots, I_{n-m}\}$.

So we have that 
\begin{equation}
\label{eqphi}
\int_{\St^{(\beta)}(n,m)}f=\left(\phi(I_1,\cdots,I_{n-m})f\right)(0),
\end{equation}
for an arbitrary polynomial $f$ and some formal power series $\phi$. Let $f\in \FS$, then by definition there are polynomials $f_j$ of degree $j$ such that $\{f_j, j\in \mN\}$ converges uniformly to $f$ on $\St^{(\beta)}(n,m)$. Hence we have
\begin{equation}
\int_{\St^{(\beta)}(n,m)}f=\lim_{j\to\infty}\int_{\St^{(\beta)}(n,m)}f_j=\lim_{j\to \infty}\left(\phi^{(j)}(I_1,\cdots,I_{n-m})f\right)(0),
\end{equation}
where $\phi^{(j)}(I_1,\cdots,I_n)$ is the polynomial given by taking all terms in $\phi(\lambda^2 I_1,\cdots,\lambda^{2n-2m}I_{n-m})$ which are of degree $j$ or lower in $\lambda$ and then setting $\lambda=1$. This implies that equation \eqref{eqphi} also holds for $f\in \FS$.

Considering $\exp\left[-\imath(\tr B^\dagger X+\tr X^\dagger B)/(2\gamma)\right]\in \FS$ for arbitrary $B\in \Gl^{(\beta)}(n,m)$, we have the equality
\begin{equation}
\phi( I_1(B),\cdots, I_{n-m}(B))= \int_{\St^{(\beta)}(n,m)}d\mu(X)\exp\left[-\frac{\imath}{2\gamma}(\tr B^\dagger X+\tr X^\dagger B)\right]=\widehat{\Psi}^{(\beta)}_{n,0} ( B),
\end{equation}
 as a formal power series in $B$ proving that $\phi=\Psi^{(\beta)}_n$.
\end{proof}

By construction we have the equality
\begin{equation}
\Psi_{n,m}^{(\beta)}(x_1,\cdots,x_{n-m})= \Psi_{n,0}^{(\beta)}(x_1,\cdots,x_{n-m},0,\cdots,0).
\end{equation}
It therefore remains to calculate the formal power series $\Psi_{n,0}^{(\beta)}$. The intuitive results in section \ref{sketch} already deliver some insight into possible solutions. In section \ref{sec2} we will prove that the formulae suggested in subsection \ref{newPiz} for $m=n-2$ are correct, i.e. 
\begin{equation}\Psi_{n,0}^{(\beta)}(x_1,x_2,0,\cdots,0)=\sum_{j=0}^{\infty}\frac{\Gamma(\beta n/2)}{4^j \Gamma(j+\beta n/2)}\sum_{l=0}^{\lfloor j/2\rfloor}\frac{\Gamma[\beta(n-1)/2]}{\Gamma[l+\beta (n-1)/2]}\frac{x_1^{j-2l}}{(j-2l)!}\frac{x_2^l}{l!}.
\end{equation}

We conclude this section with a useful characterization of the undetermined power series $\Psi^{(\beta)}_{n,m}$. Therefore we study the space of polynomials on the real space $\Gl^{(\beta)}(n,m)$ contained in $X^\dagger X$ as a $\rmU^{(\beta)}(n-m)$-module, or more precisely as a module of the corresponding Lie algebra. For a module $M$, we use the notation $M\odot M\cong \odot^2 M$, $\,\,M^\ast$, $\,\,M\wedge M\cong \wedge^2 M$ for respectively the symmetric tensor product, the dual module and the anti-symmetric tensor product.

\begin{lemma}
\label{decompcond}
The complexification of the real space of polynomials on $\Gl^{(\beta)}(n,m)$ corresponding to the real degrees of freedom in the indeterminate (complex) matrix $X^\dagger X$ as a $\mathfrak{g}^{(\beta)}_{n-m}$-module is isomorphic to
\begin{equation}
 \left\{\begin{array}{cl} V\odot V \,\,& \mbox{ with $V$ the tautological $\mathfrak{so}(n-m;\mC)$-module,  $\beta=1$}, \\  V\otimes V^\ast & \mbox{ with $V$ the tautological $\mathfrak{gl}(n-m;\mC)$-module,  $\beta=2$}, \\V\wedge V \,\,& \mbox{ with $V$ the tautological $\mathfrak{sp}(2n-2m;\mC)$-module,  $\beta=4$} .\end{array}\right.
\end{equation}
This module decomposes into the direct sum of two simple modules if $m<n-1$ and is simple for $m=n-1$. One is isomorphic to the trivial module and is realised as $\tr X^\dagger X$, the other is the kernel of the map $X^\dagger X\to \eins_{\gamma (n-m)}$ 
\end{lemma}

\begin{proof}
In this proof we consider the most complicated scenario of square matrices, i.e. $m=0$. This shortens the notation but the proof does not change for the other cases.

First consider $\beta=1$. Define the columns $u_j$ of the matrix $X$ by $X=(u_1,\cdots,u_n)$. It is clear that the space  $\Span_{\mR}\{ u_i^Tu_j=u_j^Tu_i\}$ is isomorphic to $V\odot V$ for the $\mathfrak{so}(n;\mR)$ action on $\Gl^{(1)}(n,0)$ coming from the right $\SO(n)$-multiplication. The complexification follows trivially.

Now consider $\beta=2$. Again we set $X=(u_1,\cdots,u_n)$, where now the complex conjugates $u_j^\ast$ are independent column vectors. The action of $\mathfrak{gl}(n;\mC)$ on $\Gl^{(2)}(n,0)$ coming from the right $\rmU(n)$-multiplication acts on the complexification of the $n$-dimensional space $\Span_{\mR}\{u_i\}$ as the tautological representation and on the $n$-dimensional space $\Span_{\mR}\{u_i^\ast\}$ as the dual of that representation. This implies that the space \begin{equation}\left(\Span_{\mR}\{ (u_i^\ast)^Tu_j=u_j^\dagger u_i\}\right)_{\mC}\end{equation} is isomorphic to $V^\ast\otimes V$ as a $\mathfrak{gl}(n;\mC)$-module.

Finally take $\beta=4$, then we can describe $X\in\Gl^{(4)}(n,0)\hookrightarrow \mC^{2n\times 2n}$ as $X=(u_1, u_2,\cdots ,u_{2n})$, where all $u_j$ are $\mR$-linearly independent and their complex conjugate are completely determined by relation~\eqref{descimH}. The $2n$-dimensional space $\Span_{\mR}\{u_i\}$ clearly leads to the tautological $\mathfrak{sp}(2n;\mC)$-module. The matrix entries of $X^\dagger X$ then correspond to $u_i^T I u_j=-(u_j^T I u_i)$, with $I$ the anti-symmetric matrix in Eq. \eqref{I-def}, which then leads to the antisymmetric tensor product of the tautological module.

The decomposition of each of these tensor products into simple modules is standard.
\end{proof}

\begin{proposition}
\label{onlyu}
The formal power series $\Psi^{(\beta)}_{n,m}$ is the unique formal power series $\phi$ such that
\begin{equation}
\label{eqonlyu}
\left(\phi(I_1,\cdots,I_{n-m})u^2f\right)(0)=\left(\phi(I_1,\cdots,I_{n-m})f\right)(0),
\end{equation}
with $u^2:=(X^\dagger X)_{11}$ and $f$ an arbitrary polynomial on $\Gl^{(\beta)}(n,m)$.
\end{proposition}
\begin{proof}
In the spirit of the proof of Theorem \ref{PizzInv}, it suffices to prove that the proposed functional $T$ on the space of polynomials 
\begin{itemize}
\item[1)] is left $\rmU^{(\beta)}(n)$-invariant;
\item[2)] restricts to a functional on $\St^{(\beta)}(n,m)$.
\end{itemize} 
Condition 1) is obviously satisfied due to the invariance of $I_j$. For condition 2) we consider trivial complexification of all concepts. The second condition can then be expressed in a matrix identity as
\begin{equation}\label{condXX}T((X^\dagger X-\eins_{\gamma(n-m)})f)=0,
\end{equation}
for any $f$. The condition on $u^2$ is therefore a necessary condition, we show that it is also sufficient.

This is trivial for $m=n-1$, so we consider $m<n-1$. The functional $T$ is by construction also invariant under the right action of $\rmU^{(\beta)}(n-m)$. This action (or just a consideration by the symmetry $I_j$ under permuting columns of $X$) immediately shows that the condition for $u^2$ implies in general that
\begin{equation}T((X^\dagger X)_{ii}f)=T(f).
\end{equation}
for any $1\le i\le \gamma(n-m)$. This is the diagonal part of condition \eqref{condXX}. Furthermore, for $p:=(X^\dagger X)_{11}-(X^\dagger X)_{\gamma 2,\gamma 2}$, which is a non-trivial element of the kernel of $X^\dagger X\to \eins_{\gamma (n-m)}$, we have $T(pf)=0$. As this kernel constitutes a simple module for $\Lii_{n-m}$ by Lemma \ref{decompcond}, complimentary to the simple module $\sum_{i=1}^n(X^\dagger X)_{\gamma i,\gamma i}$, it follows immediately that condition \eqref{condXX} is satisfied.
\end{proof}

The characterization of $\Psi^{(\beta)}_{n,m}$ in proposition \ref{onlyu} concentrates the quest for Pizzetti formulae for Stiefel manifolds into one very innocent looking condition \eqref{eqonlyu}. In appendix \ref{secprop} we prove that  this condition is satisfied for the formulae for $m=n-2$ obtained in appendix \ref{app0}. However, already there, this calculation is all but straightforward.

\subsection{The Howe duality associated to a Stiefel manifold}

We consider the real space $\Gl^{(\beta)}(n,m)\cong \mR^{\beta n(n-m)}$ as a module for $\rmU^{(\beta)}(n)$ through left multiplication. 

\begin{lemma}
\label{modstr}
We have the following isomorphisms of $\Li$-modules:
\begin{equation}
\Gl^{(\beta)}(n,m)_{\mC} := \left\{\begin{array}{cl} V^{\oplus (n-m)}, & \beta=1, \\ V^{\oplus (n-m)}\oplus (V^{\ast})^{\oplus (n-m)}, & \beta=2, \\ V^{\oplus 2n-2m}, & \beta=4.\end{array}\right.
\end{equation}
with $V\cong\mC^{\gamma n}$ the tautological module of $\Li$.
\end{lemma}
\begin{proof} 
This can be proved identically as the corresponding statements in the proof of Lemma \ref{decompcond}.
\end{proof}

As in \cite{Howe} we introduce the complexification of the Lie algebra of quadratic differential operators on $\mR^{\beta n(n-m)}\cong \Gl^{(\beta)}(n,m)$, which is isomorphic to $\mathfrak{sp}(2\beta n(n-m);\mC)$. This Lie algebra has a $\mZ$-gradation
\begin{equation}
\label{paradecomp}
\mathfrak{sp}(2\beta n(n-m);\mC)=\fu^-\oplus \mathfrak{gl}(\beta n(n-m);\mC)\oplus \fu^+,\end{equation}
where in the realization mentioned above $\mathfrak{gl}(\beta n(n-m);\mC)$ is given by the differential operators of degree zero, $\fu^-$ by those of degree plus 2 and $\fu^+$ by those of degree minus 2.

We denote the Lie algebra, which is the centralizer of $\Li$ in $\mathfrak{sp}(2\beta n(n-m);\mC)$, by $\Gamma^{(\beta)}_{n,m} $ and furthermore set $\fl^{(\beta)}_{n,m}:= \mathfrak{gl}(\beta n(n-m);\mC)\cap \Gamma^{(\beta)}_{n,m}$, the differential operators in $\Gamma^{(\beta)}_{n,m} $ of zero degree. We denote by $\Gamma^{(\beta)}_{n,m}=\fv^-\oplus \fl^{(\beta)}_{n,m}\oplus\fv^+ $ the $\mZ$-gradation inherited from Eq.~\eqref{paradecomp}.

\begin{theorem}
\label{HoDu}
Excluding the case $(\beta,n,m)=(1,2,0)$, we have
\begin{equation}
\Gamma^{(\beta)}_{n,m} := \left\{\begin{array}{cl} \mathfrak{sp}(2n-2m;\mC) & \\ \mathfrak{gl}(2n-2m;\mC) &  \\ \mathfrak{so}(4n-4m;\mC) & \end{array}\right.\mbox{and }\quad \fl^{(\beta)}_{n,m} := \left\{\begin{array}{cl} \mathfrak{gl}(n-m;\mC) &  \\ \mathfrak{gl}(n-m;\mC)\oplus  \mathfrak{gl}(n-m;\mC) &\\ \mathfrak{gl}(2n-2m;\mC) & \end{array}\right.\mbox{for }\quad \beta=\left\{\begin{array}{cl} 1 \\ 2 \\ 4.\end{array}\right.
\label{algebras2}
\end{equation}
Furthermore the commutative algebra $\fv^+$ is spanned by the complexification of the polynomials on the real space $\Gl^{(\beta)}(n,m)$ that are in the matrix $X^\dagger X$. This implies in particular that
\begin{equation}
\dim_{\mC}\fv^+\,=\,\dim_{\mR} \Gl^{(\beta)}(n,m)-\dim_{\mR}\St^{(\beta)}(n,m).
\end{equation}
\end{theorem}
\begin{proof}
Computing the $\Li$-invariants in $\mathfrak{sp}(2\beta n(n-m);\mC)$, corresponds to computing the invariants in 
\begin{equation}\odot^2\Gl^{(\beta)}(n,m)_{\mC} \;\;\;\bigoplus\;\;\;\left(\Gl^{(\beta)}(n,m)_{\mC}\right)^\ast\otimes \Gl^{(\beta)}(n,m)_{\mC}\;\;\;\bigoplus\;\;\;\odot^2 \left(\Gl^{(\beta)}(n,m)_{\mC}\right)^\ast,\end{equation}
where the decomposition corresponds to the $\mZ$-gradation in Eq. \eqref{paradecomp}. The results then follow immediately from Lemma \ref{modstr}.
\end{proof}

Note that we have logical embeddings $\Lii_{n-m}\hookrightarrow \fl^{(\beta)}_{n,m}\hookrightarrow \Gamma^{(\beta)}_{n,m}$, since in particular the right action of $\rmU^{(\beta)}(n-m)$ on $\Gl^{(\beta)}(n,m)$ commutes with the left action of $\rmU^{(\beta)}(n)$.


In this setup we can derive an alternative interpretation for the Pizzetti formula in Theorem \ref{PizzInv}. As noted in Theorem \ref{HoDu}, the elements of $\fv^-$ are the quadratic polynomials on $\Gl^{(\beta)}(n,m)$ which define $\St^{(\beta)}(n,m)$. Hence the pull-back $\iota^\sharp$ of the embedding $\iota:\St^{(\beta)}(n,m)\hookrightarrow \Gl^{(\beta)}(n,m)$ evaluated on $\fv^-$ gives precisely the map $X^\dagger X\to \eins$.

\begin{theorem}
\label{interpHowe}
(i) The unique (up to multiplicative constant) linear functional $T$ on $\mC[\Gl^{(\beta)}(n,m)_{\mC}]$ satisfying
\begin{itemize}
\item $T$ is left $\rmU^{(\beta)}(n)$-invariant
\item $T(A\cdot)=\iota^\sharp(A)T(\cdot)$ for any $A\in \fv^-$
\end{itemize}
is given by $\int_{\St^{(\beta)}(n,m)}$.

\noindent
(ii) The unique (up to multiplicative constant) $\rmU^{(\beta)}(n)$-invariant linear functional on $\left(\mC\left[\Gl^{(\beta)}(n,m)\right]\right)^{\fv^+}$ is given by $\int_{\St^{(\beta)}(n,m)}$.
\end{theorem}

\begin{proof}
We prove part $(i)$. By theorem \ref{HoDu}, the second condition implies precisely that the functional on polynomials on $\Gl^{(\beta)}(n,m)$, depends only on the restriction of the polynomial to $\St^{(\beta)}(n,m)$. The unique functional which satisfies both conditions corresponds to the $\rmU^{(\beta)}(n)$-invariant integration over $\St^{(\beta)}(n,m)$.

For part $(ii)$ we take into account that $\mC\left[\Gl^{(\beta)}(n,m)_{\mC}\right]$ as a $\Gamma^{(\beta)}_{n,m}$-module decomposes into a direct sum of simple highest weight modules, see e.g. Theorem 8 in \cite{Howe}. In combination with the $\mZ$-gradation of $\Gamma^{(\beta)}_{n,m}$, this implies in particular that
\begin{equation}
\mC\left[\Gl^{(\beta)}(n,m)_{\mC}\right]=U(\fv^-)\left(\mC\left[\Gl^{(\beta)}(n,m)_{\mC}\right]\right)^{\fv^+}.
\end{equation} 
This consideration and the fact that the elements in the realization of $u(\fv^-)$ are by definition $\rmU^{(\beta)}$-invariant immediately yields a one-to-one correspondence between the functionals satisfying the conditions in $(i)$ and those in $(ii)$, concluding the proof.
\end{proof}

\subsection{Example 1, $m=n-1$: Harmonic analysis}

If we set $m=n-1$, then we get $\Gl^{(\beta)}(n,n-1)\cong \mR^{\beta n}$ and $\St^{(\beta)}(n,n-1)\cong \mS^{\beta n-1}$ canonically embedded.

Even though the Howe dual pairs are quire different,
\begin{equation}
\left\{\begin{array}{cl} \mathfrak{so}(n;\mC)\times \mathfrak{sp}(2;\mC)&\subset\mathfrak{sp}(2n;\mC) \\  \mathfrak{gl}(n;\mC)\times \mathfrak{gl}(2;\mC)&\subset\mathfrak{sp}(4n;\mC) \\  \mathfrak{sp}(2n;\mC)\times \mathfrak{so}(4;\mC)&\subset\mathfrak{sp}(8n;\mC),\end{array}\right.
\end{equation}
the point is that in each case, $\fv^+$ and $\fv^-$ are one-dimensional and spanned by respectively the Laplacian and norm squared on $\mR^{\beta n}$. Therefore, these three different types of Howe dualities exhibit the same features as in Theorem \ref{interpHowe}.

\subsection{Example 2, $\beta=1$: Harmonic analysis in multiple sets of variables}

Now we consider the case $\St^{(1)}(n,n-k)$ for $1\le k \le n$, then $\Gl^{(1)}(n,n-k)\cong \mR^{n\times k}$. We define $k$ $n$-dimensional variables by putting $X=(u_1,\cdots,u_k)$. By Theorem \ref{HoDu} we have
\begin{equation}
\fv^-=\Span\{\langle \nabla_{u_i},\nabla_{u_j}\rangle\,|\, 1\le i,j\le k\}\quad\mbox{and}\quad \fv^+=\Span\{\langle {u_i},{u_j}\rangle\,|\, 1\le i,j\le k\}.
\end{equation}

The null-solutions of $\fv^-$ are known as harmonic functions. Theorem \ref{interpHowe} therefore interprets the formulae in Theorem \ref{PizzInv} as either the unique $\SO(n)$-invariant functional $T$ on polynomials on $\mR^{n\times k}$ which satisfies $T(\langle u_i,u_j\rangle\cdot)=\delta_{ij}T(\cdot)$; or as the unique $\SO(n)$-invariant functional on the space of harmonic functions on $\mR^{k\times n}$.

\section{Pizzetti-type formulae for $\St^{(\beta)}(n,n-2)$}
\label{sec2}

The main result of this section calculates the symbolical expression in Theorem \ref{PizzInv} explicitly and rigorously for $m=n-2$, confirming Eqs. \eqref{Pizzetti-beta2-2bb} and \eqref{Pizzetti-beta4-2b}. This is presented in the following theorem.
\begin{theorem}
\label{thmsum2}
For any $f\in \cA^{(\beta)}_{n,n-2}$ we have
\begin{equation}
\label{thmsum2eq}
\int_{\St^{(\beta)}(n,n-2)}f=\sum_{j=0}^{\infty}\frac{\Gamma(\beta n/2)}{4^j \Gamma(j+\beta n/2)}\sum_{l=0}^{\lfloor j/2\rfloor}\frac{\Gamma[\beta(n-1)/2]}{\Gamma[l+\beta (n-1)/2]}\left(\frac{I_1^{j-2l}}{(j-2l)!}\frac{I_2^l}{l!}f\right)(0).
\end{equation}
\end{theorem}
From the proof of Theorem~\ref{PizzInv} we know that it suffices to prove the result for polynomials since the extension to $\FS$ is always guaranteed. In the following we therefore only consider polynomials.

We can apply proposition \ref{onlyu} to reduce the proof to checking one condition for the expression on the right-hand side of Eq. \eqref{thmsum2eq}. The main technical calculation for this is performed in appendix~\ref{secprop}, leading to proposition \ref{prop}. All that remains to be done to prove theorem \ref{thmsum2} is therefore showing that the conditions to apply proposition \ref{prop} are satisfied.

\subsection{Pizzetti formula for $\St^{(1)}(n,n-2)$}

We denote the indeterminate matrix in $\Gl^{(1)}(n,n-2)=\mR^{n\times 2}$ by $X$ and introduce column vectors $u$ and $v$ as $X=(u,v)$. By Eq. \eqref{invariantsI} we have
\begin{equation}\label{defA}
I_1=\tr(\nabla_X^T\nabla_X)=\Delta_u+\Delta_v
\end{equation}
and 
\begin{equation}
\label{defB}I_2=\det(\nabla_X^T\nabla_X)=\frac{1}{2}\left(\left(\tr(\nabla_X^T\nabla_X)\right)^2-\tr\left(\nabla_X^T\nabla_X\right)^2\right)=\Delta_u\Delta_v-\langle \nabla_u,\nabla_v\rangle^2.
\end{equation}
We can therefore apply proposition \ref{prop} by choosing $m=n$, $k=1$ and $J^{(0)}=\eins_n$. In combination with proposition \ref{onlyu}, this yields the following theorem.

\begin{theorem}
\label{resStiefel}
\label{Pizzettinew}
For $f$ a polynomial on $\Gl^{(1)}(n,n-2)$ we have
\begin{equation}\label{piz-int-beta-1}
\int_{\St^{(1)}(n,n-2)}f=\sum_{j=0}^\infty\frac{\Gamma(n/2)}{4^j\Gamma(j+n/2)}\sum_{l=0}^{\lfloor j/2\rfloor}\frac{\Gamma\left[(n-1)/2\right]}{\Gamma\left[l+(n-1)/2\right]}\left(\frac{I_1^{j-2l}}{(j-2l)!}\frac{I_2^l}{l!}f\right)(0).
\end{equation}
\end{theorem}

We emphasize that for the case $n=2$ the integral~\eqref{piz-int-beta-1} corresponds to the group-invariant integration over orthogonal group ${\rm O}(2)$ and not over the special orthogonal group $\SO(2)\cong\mathbb{S}^1$. To have the integral over latter case one has to use the Pizzetti formula
\begin{equation}\label{piz-int-beta-1-special}
\int_{\SO(2)}f=\sum_{j=0}^\infty\frac{1}{4^j j!}\sum_{l=0}^{\lfloor j/2\rfloor}\frac{\sqrt{\pi}}{\Gamma\left[l+1/2\right]}\left(\frac{I_1^{j-2l}}{(j-2l)!}\frac{I_2^l}{l!}(1+u_1v_2-u_2v_1)f\right)(0).
\end{equation}
The additional factor restricts the space ${\rm O}(2)$ onto those orthogonal matrices $X\in{\rm O}(2)$ with $\det X=1$.

\subsection{Pizzetti formula for $\St^{(2)}(n,n-2)$}

We denote the indeterminate matrix by $X\in \mC^{n \times 2}\cong \Gl^{(2)}(n,n-2)$. We define the column vectors $\widetilde{u},\widetilde{v}\in\mC^{n}$ of the matrix $X=(\widetilde{u},\widetilde{v})$ and furthermore $\widetilde{u}=u^{(1)}+\imath u^{(2)}$ and $\widetilde{v}=v^{(1)}+\imath v^{(2)}$ with $u^{(1)},u^{(2)},v^{(1)},v^{(2)}\in\mR^{n}$. Finally we define $u,v\in\mR^{2n}$ by $u=\left(\begin{array}{c}u^{(1)}\\ u^{(2)}\end{array}\right)$ and $v=\left(\begin{array}{c}v^{(1)}\\ v^{(2)}\end{array}\right)$.
A direct calculation shows that Eq.~\eqref{invariantsI} yields
\begin{eqnarray}
I_1&=&\tr(\nabla_X^\dagger\nabla_X)=\Delta_{u}+\Delta_v\quad{\rm and}\nonumber\\
I_2&=&\frac{1}{2}\left(\left(\tr(\nabla_X^\dagger\nabla_X)\right)^2-\tr\left(\nabla_X^\dagger\nabla_X\right)^2\right)=\Delta_u\Delta_v-\langle \nabla_u,\nabla_v\rangle^2-\langle \nabla_u,J\nabla_v\rangle^2 \quad\mbox{with }\nonumber\\
 J&=&\left(\begin{array}{cc}0 & \eins_n\\ -\eins_n & 0\end{array}\right).\label{jdef-beta2}
\end{eqnarray}
Then proposition \ref{prop}, for $m=n$, $k=2$, $J^{(0)}=\eins_{2n}$ and $J^{(1)}=J$ and proposition \ref{onlyu} yield the following theorem.

\begin{theorem}
\label{stiefelbeta2}
For $f$ a polynomial on $\Gl^{(2)}(n,n-2)$ we have
\begin{equation}\label{theorem-b2}
\int_{\St^{(2)}(n,n-2)}f=\sum_{j=0}^{\infty}\frac{\Gamma(n)}{4^j \Gamma(j+n)}\sum_{l=0}^{\lfloor j/2\rfloor}\frac{\Gamma(n-1)}{\Gamma(l+n-1)}\left(\frac{I_1^{j-2l}}{(j-2l)!}\frac{I_2^l}{l!}f\right)(0).
\end{equation}
\end{theorem}

\subsection{Pizzetti formula for $\St^{(4)}(n,n-2)$} 

We denote the indeterminate matrix by $X\in \mH^{n\times 2}\hookrightarrow \mC^{2n \times 4}$. By Eq.~\eqref{invariantsI-beta4} we have
\begin{equation}
I_1:=\tr(\nabla_{X}^\dagger\nabla_X)\qquad \mbox{and }\qquad I_2:=\Pf(I\nabla_X^\dagger\nabla_X),
\end{equation}
with $I$ given in equation~\eqref{I-def}. We define the column vectors $\widetilde{u},\widetilde{v}\in\mH^{n}$ of the matrix $X=(\widetilde{u},\widetilde{v})$ and use Eq.~\eqref{quat-Pauli} to define
\begin{equation}
\widetilde{u}=u^{(1)}+\imath\tau_3 u^{(2)}+\imath\tau_2 u^{(3)}+\imath\tau_1  u^{(4)}\quad\mbox{and}\quad\widetilde{v}=v^{(1)}+\imath\tau_3 v^{(2)}+\imath\tau_2 v^{(3)}+\imath\tau_1 v^{(4)}
\end{equation}
 with $u^{(a)},v^{(a)}\in\mR^{n\times 1}$. Finally we define $u,v\in\mR^{4n\times 1}$ by $$u=\left(\begin{array}{c}u^{(1)}\\ u^{(2)}\\u^{(3)}\\u^{(4)}\end{array}\right)\quad \mbox{and}\quad v=\left(\begin{array}{c}v^{(1)}\\ v^{(2)}\\v^{(3)}\\ v^{(4)}\end{array}\right).$$
Then we can calculate
\begin{equation}
I_1=\Delta_{u}+\Delta_v\qquad\mbox{and }\quad I_2=\Delta_u\Delta_v-\langle \nabla_u,\nabla_v\rangle^2-\langle \nabla_u,J^{(1)}\nabla_v\rangle^2-\langle \nabla_u,J^{(2)}\nabla_v\rangle^2-\langle \nabla_u,J^{(3)}\nabla_v\rangle^2,
\end{equation}
with
\begin{equation}\label{Jquatdef}
J^{(1)}=\left(\begin{array}{cccc}0 & \eins_n&0&0\\ -\eins_n& 0&0&0\\0&0&0&-\eins_n\\0&0&\eins_n&0\end{array}\right),\,\,\, J^{(2)}=\left(\begin{array}{cccc}0 & 0&\eins_n&0\\ 0& 0&0&\eins_n\\-\eins_n&0&0&0\\0&-\eins_n&0&0\end{array}\right)\end{equation}
and
\begin{equation}\label{Jquatdef2}
J^{(3)}=J^{(1)}J^{(2)}=\left(\begin{array}{cccc}0 & 0&0&\eins_n\\ 0 & 0&-\eins_n&0\\0&\eins_n&0&0\\-\eins_n&0&0&0\end{array}\right).
\end{equation}
Thus proposition \ref{prop} for $m=n$, $k=2$, $J^{(0)}=\eins_{2n}$ and $J^{(1)}=J$ and proposition \ref{onlyu} yield the following result.

\begin{theorem}
\label{stiefelbeta4}
For $f$ a polynomial on $\Gl^{(4)}(n,n-2)$ we have
\begin{equation}\label{theorem-b4}
\int_{\St^{(4)}(n,n-2)}f=\sum_{j=0}^{\infty}\frac{\Gamma(2n)}{4^j \Gamma(j+2n)}\sum_{l=0}^{\lfloor j/2\rfloor}\frac{\Gamma(2n-2)}{\Gamma(l+2n-2)}\left(\frac{I_1^{j-2l}}{(j-2l)!}\frac{I_2^l}{l!}f\right)(0).
\end{equation}
\end{theorem}

\begin{remark}
Note that the matrices $J^{(1)}$, $J^{(2)}$ and $J^{(3)}$ satisfy the relations $J^{(1)}J^{(2)}=J^{(3)}$, $J^{(2)}J^{(3)}=J^{(1)}$ and $J^{(3)}J^{(1)}=J^{(2)}$ as well as $(J^{(i)})^2=-\eins_{4n}$. This implies that these matrices generate an algebra isomorphic to $\mH$, see also remark \ref{remCli}.
\end{remark}

\section{Integration of Schwartz functions over $\rmU^{(\beta)}(n)$}\label{sec:Schwartz}

In the previous sections we proved the Pizzetti formulae of the Stiefel manifold $\St^{(\beta)}(n,n-2)$ on functions with a certain analyticity condition. What remains are two generalizations. First we want to go over to integrations over Schwartz functions, where $\cS(M)$ is the set of Schwartz functions on the manifold $M$. Second we want to find Pizzetti formulae for general Stiefel manifolds $\St^{(\beta)}(n,m)$. 

As noted in section \ref{sketch}, the classical Pizzetti formula \eqref{classPizzetti} for $\St^{(\beta)}(n,n-1)$ is applicable for Schwartz functions on $\Gl^{(\beta)}(n,n-1)\cong \mR^{\beta n}$. Let us first concentrate on the case $\beta=1$. For a function $f\in\cS(\mR^{n\times k})$ we define the Fourier transform in the $k$th $n$-dimensional variable as
\begin{equation}
\cF_{n,k}^{(1)}[f](u_1,\dots,u_{k-1},y_k):=\frac{1}{(2\pi)^n}\int_{\mR^n}d[u_k]\exp[\imath\langle y_k,u_k\rangle]f(u_1,\cdots,u_k)\ \in\cS(\mR^{n\times k}).
\end{equation}
Subsequently we define a transform from $\cS(\mR^{n\times k})$ to $\cS(\mR^{n\times (k-1)})$ as
\begin{equation}\label{T-def}
T_{n,k}^{(1)}[f](u_1,\cdots,u_{k-1}):=\int_{\mR^n}d[y_k]\, \Psi_{(n-k-1)/2}\left(\sqrt{y_k^2-\sum_{i=1}^{k-1}\langle u_i,y_k\rangle^2}\right)\,\cF_{n,k}^{(1)}[f](u_1,\cdots,u_{k-1},y_k).
\end{equation}
Recall that $\Psi_p$ is the renormalized Bessel function~\eqref{renormBessel}.

As before we can rewrite this into the following Pizzetti-like expression:
\begin{eqnarray}
\nonumber
T_{n,k}^{(1)}[f](u_1,\cdots,u_{k-1})&=&\sum_{j=0}^\infty\frac{\Gamma\left[(n-k+1)/2\right]}{4^j j!\Gamma\left[j+(n-k+1)/2\right]}\left(\left(\Delta_{u_k}-\sum_{i=1}^{k-1}\langle u_i,\nabla_{u_k}\rangle^2\right)^j f\right)(u_1,\cdots,u_{k-1},0)\\
\label{T-reform}
&=&\left(\Psi_{(n-k-1)/2}\left(\sqrt{\sum_{i=1}^{k-1}\langle u_i,\nabla_{u_k}\rangle^2-\Delta_{u_k}}\right)f\right)(u_1,\cdots,u_{k-1},0)
\end{eqnarray}

The Pizzetti formula for $\St^{(1)}(n,n-k)=\SO(n)/\SO(n-k)$ is stated in the following theorem.
\begin{theorem}
For $f\in\cS(\mR^{n\times k})$ with $k<n$, the $\SO(n)$-invariant integration over the Stiefel manifold $\St^{(1)}(n,n-k)$ (canonically embedded in $\mR^{n\times k}$) is given by
\begin{equation}\label{Pizzetti-ful-beta1}
\int_{\St^{(1)}(n,n-k)}f=T_{n,1}^{(1)}\circ T_{n,2}^{(1)}\circ \cdots \circ T_{n,k}^{(1)} [f].
\end{equation}
\end{theorem}

Note that the idea behind this theorem is the telescopic factorization of the Stiefel manifold
\begin{eqnarray}
\St^{(1)}(n,n-k)&=&\SO(n)/\SO(n-k)\nonumber\\
&\cong&\SO(n)/\SO(n-1)\times\SO(n-1)/\SO(n-2)\times\cdots\times\SO(n-k+1)/\SO(n-k)\nonumber\\
&=&\St^{(1)}(n,n-1)\times\St^{(1)}(n-1,n-2)\times\cdots\times\St^{(1)}(n-k+1,n-k).
\end{eqnarray}
 Thereby the Stiefel manifolds in the last line are embedded in $\St^{(1)}(n,n-k)$ in a specific way. The corresponding additional conditions are the reason for the non-trivial argument of the renormalized Bessel function $\Psi_p$ in Eq.~\eqref{T-reform}.
 
\begin{proof}
We prove this statement by induction. The case $k=1$ is certainly true because $\St^{(1)}(n,n-1)\cong\mathbb{S}^{n-1}$ for which we can apply the original Pizzetti formula~\eqref{classPizzetti}, see Eq.~\eqref{connPizz}. Hence we consider the induction $k-1\to k$ where we assume that Eq.~\eqref{Pizzetti-ful-beta1} holds for an $k-1\in\mathbb{N}$. Since
\begin{equation}
\label{uk2}
T_{n,k}^{(1)}[ u_k^2 f]=-\int_{\mR^n}d[y_k]\left(\Delta_{y_k} \Psi_{(n-k-1)/2}\left(\sqrt{y_k^2-\sum_{i=1}^{k-1}\langle u_i,y_k\rangle^2}\right)\right)\cF_{n,k}^{(1)}[f](u_1,\cdots,u_{k-1},y_k),
\end{equation} 
we calculate
\begin{eqnarray}
&&\Delta_{y_k} \Psi_{(n-k-1)/2}\left(\sqrt{y_k^2-\sum_{i=1}^{k-1}\langle u_i,y_k\rangle^2}\right)\nonumber\\
&=& \left(\frac{\Psi''_{(n-k-1)/2}\left(\sqrt{y_k^2-\sum_{i=1}^{k-1}\langle u_i,y_k\rangle^2}\right)}{y_k^2-\sum_{i=1}^{k-1}\langle u_i,y_k\rangle^2}-\frac{\Psi'_{(n-k-1)/2}\left(\sqrt{y_k^2-\sum_{i=1}^{k-1}\langle u_i,y_k\rangle^2}\right)}{(y_k^2-\sum_{i=1}^{k-1}\langle u_i,y_k\rangle^2)^{3/2}}\right)\nonumber\\
&&\times\left(y_k^2-2\sum_{i=1}^{k-1}\langle u_i,y_k\rangle^2+\sum_{i,j=1}^{k-1}\langle u_i,u_j\rangle\langle u_i,y_k\rangle\langle u_j,y_k\rangle\right)\nonumber\\
&&+\frac{\Psi'_{(n-k+1)/2-1}\left(\sqrt{y_k^2-\sum_{i=1}^{k-1}\langle u_i,y_k\rangle^2}\right)}{\sqrt{y_k^2-\sum_{i=1}^{k-1}\langle u_i,y_k\rangle^2}}\left(n-\sum_{i=1}^{k-1}u_i^2\right).
\end{eqnarray}
Since we study the evaluation of $T_{n,1}^{(1)}\circ\cdots\circ T_{n,k-1}^{(1)}$ on the function \eqref{uk2}, the induction step implies that $\langle u_i,u_j\rangle$ can be replaced by the Kronecker symbol $\delta_{ij}$, which yields
\begin{eqnarray}
&&\Psi''_{(n-k-1)/2}\left(\sqrt{y_k^2-\sum_{i=1}^{k-1}\langle u_i,y_k\rangle^2}\right)+\left(n-k\right)\frac{\Psi'_{(n-k-1)/2}\left(\sqrt{y_k^2-\sum_{i=1}^{k-1}\langle u_i,y_k\rangle^2}\right)}{\sqrt{y_k^2-\sum_{i=1}^{k-1}\langle u_i,y_k\rangle^2}}\nonumber\\
&=&-\Psi_{(n-k-1)/2}\left(\sqrt{y_k^2-\sum_{i=1}^{k-1}\langle u_i,y_k\rangle^2}\right),
\end{eqnarray}
which follows from the defining differential equation of the Bessel function. Therefore we obtain 
\begin{equation}
T_{n,1}^{(1)}\circ\cdots\circ T_{n,k}^{(1)}[u_k^2f]=T_{n,1}^{(1)}\circ\cdots\circ T_{n,k}^{(1)}[f].
\end{equation}
 Similar to this calculation, one can prove $T_{n,1}^{(1)}\circ \cdots\circ T_{n,k}^{(1)}[\langle u_i,u_k\rangle f]=0$. But this follows also immediately from the idea in proposition \ref{onlyu}. Hence, the functional only depends on the restriction of the function to $\St^{(1)}(n,n-k)\subset\mR^{n\times k}$. 

What remains is to prove that $T_{n,k}^{(1)}: \cS(\mR^{n\times k})\to \cS(\mR^{n\times(k-1)})$ is $\SO(n)$-equivariant, from which the $\SO(n)$ invariance of $T_{n,1}^{(1)}\circ\cdots\circ T_{n,k}^{(1)}$ follows. Therefore the functional must be the unique invariant integration.

Take $A\in \SO(n)$
and define $(Af)(u_1,\cdots,u_k)=f(A^{-1}u_1,\cdots, A^{-1}u_k)$. Then,
\begin{eqnarray}
&&T_{n,k}^{(1)}[Af](u_1,\cdots,u_{k-1})\nonumber\\
&=&\int_{\mR^n}d[y_k] \Psi_{(n-k-1)/2}\left(\sqrt{y_k^2-\sum_{i=1}^{k-1}\langle u_i,y_k\rangle^2}\right)\cF_{n,k}^{(1)}[f](A^{-1}u_1,\cdots,A^{-1}u_{k-1},A^{-1}y_k)\nonumber\\
&\overset{y_k\to Ay_k}{=}&\int_{\mR^n}d[y_k] \Psi_{(n-k-1)/2}\left(\sqrt{y_k^2-\sum_{i=1}^{k-1}\langle A^{-1}u_i,y_k\rangle^2}\right)\cF_{n,k}^{(1)}[f](A^{-1}u_1,\cdots,A^{-1}u_{k-1},y_k)\nonumber\\
&=&T_{n,k}^{(1)}[f](A^{-1}u_1,\cdots,A^{-1}u_{k-1})=(AT_{n,k}^{(1)}[f])(u_1,\cdots,u_{k-1})
\end{eqnarray}
holds, which concludes the proof.
\end{proof}

This proof can be readily generalized to arbitrary Dyson index $\beta=1,2,4$ such that we have the following theorem.

\begin{theorem}\label{recursion-theorem}
For $f\in\cS(\mR^{\beta n\times k})$, the $\rmU^{(\beta)}(n)$-invariant integration over the Stiefel manifold $\St^{(\beta)}(n,n-k)$ (canonically embedded in $\mR^{\beta n\times k}$) is given by
\begin{equation}
\int_{\St^{(\beta)}(n,n-k)}f=T_{n,1}^{(\beta)}\circ T_{n,2}^{(\beta)}\circ \cdots \circ T_{n,k}^{(\beta)} [f],
\end{equation}
with
\begin{equation}
T_{n,k}^{(\beta)}[f](u_1,\cdots u_{k-1})=\left(\Psi_{\beta(n-k+1)/2-1}\left(\sqrt{\sum_{i=1}^{k-1}\sum_{l=0}^{\beta-1}\langle u_i,J^{(l)}\nabla_{u_k}\rangle^2-\Delta_{u_k}}\right)f\right)(u_1,\cdots,u_{k-1},0)
\end{equation}
with $u_j\in\mathbb{R}^{\beta n}$ (recall that $\mathbb{C}^n\cong\mathbb{R}^{2n}$ and $\mathbb{H}^n\cong\mathbb{R}^{4n}$) and $J^{(0)}=\eins_{\beta n}$ and $J^{(l\neq0)}$ given by Eqs.~\eqref{jdef-beta2} or \eqref{Jquatdef} and \eqref{Jquatdef2} according to the Dyson index. 
\end{theorem}
The proof for  $T_{n,1}^{(\beta)}\circ T_{n,2}^{(\beta)}\circ \cdots \circ T_{n,k}^{(\beta)}[u_k^2f]=T_{n,1}^{(\beta)}\circ T_{n,2}^{(\beta)}\circ \cdots \circ T_{n,k}^{(\beta)}[f]$ and for $T_{n,1}^{(\beta)}\circ T_{n,2}^{(\beta)}\circ \cdots \circ T_{n,k}^{(\beta)}[\langle u_j,u_i\rangle f]=0$ is identical to the one for the case $\beta=1$. Also the calculation for $T_{n,1}^{(\beta)}\circ T_{n,2}^{(\beta)}\circ \cdots \circ T_{n,k}^{(\beta)}[\langle u_j, J^{(l)}u_i\rangle f]=0$ with $l=1$ for $\beta=2$, see Eq.~\eqref{jdef-beta2}, and  $l=1,2,3$ for $\beta=4$, see Eq.~\eqref{Jquatdef}, \eqref{Jquatdef2} works along the same line and does not cause any problems.

We underline that the recursion presented here is similar to the ones found in Refs.~\cite{MR1893696,BergereEynard}. Also here we split the group $\rmU^{(\beta)}(n)$ in cosets $\rmU^{(\beta)}(n)/\rmU^{(\beta)}(n-1)$ which are all isomorphic to spheres. These spheres can be understood as the column vectors which were split off in the fundamental representation of $\rmU(N)$ presented in Ref.~\cite{MR1893696} where the connection to Gelfand-Zeitlin\footnote{In other transliterations Zeitlin was written as Cetlin, Zetlin, Tzetlin or Tsetlin. However the origin of the Cyrillic spelling {\cyr Ce{\u i}tlin} supports our choice. } coordinates were pointed out.

\section{Itzykson-Zuber integral for $\SO(4)/[\SO(2)\times\SO(2)]$}\label{appl}

Now we apply our results in section \ref{sec2} to the integral~\eqref{ItzZub}. In the first step we extend the integral over the coset  $\SO(4)/[\SO(2)\times\SO(2)]$ to the Stiefel manifold $\St^{(1)}(4,2)$ such that we consider the integral
\begin{eqnarray}\label{ItzZub-sec6}
 I(H)&=&\int_{\St^{(1)}(4,2)}d\mu(X)\exp\left[-2\tr XX^T H\right]
\end{eqnarray}
with $H={\rm diag}(E_1,E_2,E_3,E_4)$. As $\exp\left[-2\tr XX^T H\right]\in\cA^{(1)}_{4,2}$, we can apply the formula in theorem~\ref{resStiefel},
\begin{equation}
I(H)=\sum_{j=0}^\infty\frac{1}{4^j(j+1)!}\sum_{l=0}^{\lfloor j/2\rfloor}\frac{\Gamma(3/2)}{\Gamma(l+3/2)}\left(\frac{(\tr\nabla_X^T\nabla_X)^{j-2l}}{(j-2l)!}\frac{\det(\nabla_X^T\nabla_X)^l}{l!}\exp\left[-2\tr XX^T H\right]\right)(0).
\end{equation}

To evaluate this expression we will, in the spirit of section \ref{sketch}, use the Fourier transform. To have a Schwartz function, which implies in particular that the Fourier transform exists, we choose H to be positive definite and relax this assumption at the end of our calculation since the result is analytic in $H$. This allows to rewrite $I(H)$ as
\begin{equation}
\frac{1}{(8\pi)^4\det H}\int_{\mathbb{R}^{4\times 2}} d[B]\sum_{j,l|l\le \lfloor j/2\rfloor}\frac{\Gamma(3/2)(\tr B^TB)^{j-2l}\det(B^TB)^l}{4^j(j+1)!\Gamma(l+3/2)(j-2l)!l!}\exp\left[-\frac{1}{8}\tr B^TB H^{-1}\right],
\end{equation}
where we underline that the integrals and sums are absolutely convergent such that we can exchange them without any problems. As for any $M\in \mR^{2\times 2}$, we have
\begin{eqnarray}
\int_0^{2\pi}\frac{d\varphi}{2\pi}e^{-\imath j\varphi}{\det}^{1/2+j}(\eins_2+e^{\imath\varphi}M)&=&\int_0^{2\pi}\frac{d\varphi}{2\pi}e^{-\imath j\varphi}\left(1+e^{\imath \phi}\tr M+e^{2\imath\phi}\det M)\right)^{j+1/2}\\
\nonumber
&=&\sum_{l=0}^{\lfloor j/2\rfloor}\frac{\Gamma(j+3/2)}{\Gamma(l+3/2)(j-2l)!l!}(\tr M)^{j-2l}(\det M)^{l},
\end{eqnarray}
we can apply this to $B^T B$ to obtain
\begin{eqnarray}
 I(H)&=&\frac{2}{(8\pi)^4\det H}\sum_{j=0}^\infty\frac{1}{(2j+2)!}\int_0^{2\pi} \frac{d\varphi}{2\pi}e^{-\imath j\varphi}\nonumber\\
&&\times\int_{\mathbb{R}^{4\times 2}} d[B]{\det}^{1/2+j}(\eins_4-e^{\imath\varphi}  BB^T)\exp\left[-\frac{1}{8}\tr BB^TH^{-1}\right].\label{ItzZub-sec6.a}
\end{eqnarray}
Here we used $\det(\eins_{2}-e^{\imath\varphi}B^TB)=\det(\eins_{4}-e^{\imath\varphi}BB^T).$

In the next step we replace the matrix $H^{-1}$ in the exponential function by a full $4\times4$ real symmetric matrix  $K=\{K_{ab}=K_{ba}\}$. This step is legitimate because the integral is invariant under $B\to OB$ with $O\in\SO(4)$. Defining the gradient in the matrix $K$,
\begin{equation}\label{graddef-K}
 \nabla_K=\left(\begin{array}{cccc} \displaystyle\frac{\partial}{\partial K_{11}} & \displaystyle\frac{1}{2}\frac{\partial}{\partial K_{12}} & \displaystyle\frac{1}{2}\frac{\partial}{\partial K_{13}} & \displaystyle\frac{1}{2}\frac{\partial}{\partial K_{14}} \\ \displaystyle\frac{1}{2}\frac{\partial}{\partial K_{12}} & \displaystyle\frac{\partial}{\partial K_{22}} & \displaystyle\frac{1}{2}\frac{\partial}{\partial K_{23}} & \displaystyle\frac{1}{2}\frac{\partial}{\partial K_{23}} \\ \displaystyle\frac{1}{2}\frac{\partial}{\partial K_{13}} & \displaystyle\frac{1}{2}\frac{\partial}{\partial K_{23}} & \displaystyle\frac{\partial}{\partial K_{33}} & \displaystyle\frac{1}{2}\frac{\partial}{\partial K_{34}} \\ \displaystyle\frac{1}{2}\frac{\partial}{\partial K_{14}} & \displaystyle\frac{1}{2}\frac{\partial}{\partial K_{24}} & \displaystyle\frac{1}{2}\frac{\partial}{\partial K_{34}} & \displaystyle\frac{\partial}{\partial K_{44}} \end{array}\right),
\end{equation}
we can replace the polynomial in $BB^T$ in front of the exponential function by this gradient. Then we integrate over $B$ and rescale $e^{\imath\varphi}\to e^{\imath\varphi}/8$ such that we find
\begin{eqnarray}
 I(H)&=&\frac{2}{\det H}\sum_{j=0}^\infty\frac{8^{j}}{(2j+2)!}\int_0^{2\pi} \frac{d\varphi}{2\pi}e^{-\imath j\varphi}\left.{\det}^{1/2+j}(e^{\imath\varphi}  \nabla_K+\eins_4)\frac{1}{{\det}K}\right|_{K=H^{-1}}.\label{ItzZub-sec6.b}
\end{eqnarray}
The prefactors $1/2$ in the off-diagonal elements of the gradient are important because of the symmetric structure of $K$ and guarantees the invariance $K\to OKO^T\Rightarrow \nabla_K\to O\nabla_KO^T$ with $O\in\SO(4)$. This invariance also allows us to diagonalize $K$ and to drop the angular derivatives in the differential operator such that it depends on the eigenvalues, only. Then we can replace the eigenvalues of $K$ by those of $H^{-1}$. For this purpose we express the integral over $\varphi$ in terms of four invariant differential operators,
\begin{eqnarray}\label{contour-integral}
 &&\int_0^{2\pi} \frac{d\varphi}{2\pi}e^{-\imath j\varphi}{\det}^{1/2+j}(e^{\imath\varphi}  \nabla_K+\eins_4)\\
 &=&\sum_{l=1}^{\lfloor j/4\rfloor}\sum_{k=1}^{\lfloor (j-4l)/3\rfloor}\sum_{p=1}^{\lfloor (j-4l-3k)/2\rfloor}\frac{\Gamma(3/2+j)}{l!k!p!\Gamma(3/2+3l+2k+p)} \widehat{I}_4^{j-4l-3k-2p} \widehat{I}_3^p \widehat{I}_2^k\widehat{I}_1^l\nonumber
\end{eqnarray}
with
\begin{eqnarray}\label{differential-op}
 \widehat{I}_1=\det\nabla_K,\ \widehat{I}_2=\frac{1}{6}(\tr^3\nabla_K-3\Delta_K\tr \nabla_K+2\tr\nabla_K^3) ,\ \widehat{I}_3=\frac{1}{2}(\tr^2\nabla_K-\Delta_K),\ \widehat{I}_4=\tr\nabla_K.
\end{eqnarray}
These four differential operators are the coefficients of the characteristic polynomial $\det(\nabla_K-\lambda\eins_4)$. After diagonalizing $K=OkO^T$ with $k={\rm diag}(k_1,k_2,k_3,k_4)$ and $O\in\SO(4)$ this characteristic polynomial can be written as
\begin{eqnarray}
 \det(\nabla_K-\lambda\eins_4)\rightarrow \widehat{D}(\lambda)&=&\frac{1}{16\Delta_{4}(k)}\sum_{r\in\{0,1\}^4}\prod\limits_{j=1}^4\left(\frac{\partial}{\partial k_a}-\lambda\right)^{r_a}\Delta_4(k)\prod\limits_{j=1}^4\left(\frac{\partial}{\partial k_a}-\lambda\right)^{1-r_a}\nonumber\\
 &=&\frac{1}{\Delta_4(x)}\det\left[k_a^{4-b}\left(\frac{\partial}{\partial k_a}+\frac{4-b}{2}\frac{1}{k_a}-\lambda\right)\right]_{1\leq a,b\leq 4}.\label{Sekiguchi}
\end{eqnarray}
This operator is a Sekiguchi-like differential operator and was derived by one of the authors in Ref.~\cite{KGG}. The expansion in $\lambda$ yields a set of operators building an algebraic basis of Casimir operators only expressed in the eigenvalues of $K$, in particular they commute with each other. This reads for the operators~\eqref{differential-op}
\begin{eqnarray}
\widehat{I}_1 &\rightarrow& \widehat{D}_1=\widehat{D}(0),\nonumber\\
\widehat{I}_2 &\rightarrow& \widehat{D}_2=-\frac{\partial \widehat{D}}{\partial\lambda}(0)=[\widehat{D}(0),\tr k]=\widehat{D}(0)\tr k-\tr k \widehat{D}(0),\nonumber\\
\widehat{I}_3 &\rightarrow& \widehat{D}_3=\sum_{1\leq a<b\leq 4}\left[\frac{\partial^2}{\partial k_a\partial k_b}-\frac{1}{2}\frac{1}{k_a-k_b}\left(\frac{\partial}{\partial k_a}-\frac{\partial}{\partial k_b}\right)\right],\nonumber\\
\widehat{I}_4 &\rightarrow& \widehat{D}_4=\sum_{a=1}^4\frac{\partial}{\partial k_a}.\label{differential-op.b}
\end{eqnarray}
Indeed the application of $\widehat{D}_j$ on ${\det}^{-1}k$ is non-trivial but surprisingly the most complicated operator $\widehat{D}_1$ has a simple action,
\begin{equation}
 \widehat{D}_1{\det}^a k=\frac{2a(2a+1)(2a+2)(2a+3)}{16} {\det}^{a-1} k\ \Rightarrow\ \widehat{D}_1{\det}^{-1} k=0.
\end{equation}
Therefore one sum vanishes and the integral~\eqref{ItzZub-sec6.b} reads
\begin{eqnarray}
 I(H)&=&\frac{2}{\det H}\sum_{j=0}^\infty\frac{8^{j}}{(2j+2)!}\sum_{k=1}^{\lfloor j/3\rfloor}\sum_{p=1}^{\lfloor (j-3k)/2\rfloor}\frac{\Gamma(3/2+j)}{k!p!\Gamma(3/2+2k+p)}\left. \widehat{D}_4^{j-3k-2p} \widehat{D}_3^p\widehat{D}_2^k{\det}^{-1} k\right|_{k=H^{-1}}.\label{ItzZub-sec6.c}
\end{eqnarray}
One can also show that
\begin{equation}
\widehat{D}_2{\det}^{-1} k=\widehat{D}_1\tr k {\det}^{-1} k=0
\end{equation}
such that we have
\begin{eqnarray}
 I(H)&=&\frac{1}{\det H}\sum_{j=0}^\infty\sum_{p=1}^{\lfloor j/2\rfloor}\frac{2^{j+2p}}{(j+1)!(2p+1)!}\left. \widehat{D}_4^{j-2p}\widehat{D}_3^p{\det}^{-1} k\right|_{k=H^{-1}}.\label{ItzZub-sec6.d}
\end{eqnarray}
Also the action of $\widehat{D}_4^{j-2p}$ can be simply done though it does not vanish,
\begin{eqnarray}
 I(H)&=&\frac{1}{\det H}\sum_{j=0}^\infty\sum_{p=1}^{\lfloor j/2\rfloor}\frac{(-2)^{j+2p}(j+2p)!}{(j+1)!(2p+1)!}\left.\widehat{D}_3^p \frac{1}{\Delta_4(k)}\det\left[\begin{array}{c} \displaystyle\left\{ k_b^{3-a}\right\}\underset{1\leq b\leq 4}{\underset{1\leq a\leq3}{\ }} \\ \displaystyle\left\{ k_b^{2p-j-1}\right\}\underset{1\leq b\leq 4}{\ } \end{array}\right]\right|_{k=H^{-1}}.\label{ItzZub-sec6.e}
\end{eqnarray}
The remaining action of $\widehat{D}_3^p$ becomes quite cumbersome but can be readily numerically evaluated. Hence we only rewrite everything in $H={\rm diag}(E_1,E_2,E_3,E_4)$ and end up with
\begin{eqnarray}
 I(H)&=&\frac{1}{\det H}\sum_{j=0}^\infty\sum_{p=1}^{\lfloor j/2\rfloor}\frac{(-2)^{j+2p}(j+2p)!}{(j+1)!(2p+1)!}D_3^p \frac{1}{\Delta_4(H)}\det\left[\begin{array}{c} \displaystyle\left\{ E_b^{a}\right\}\underset{1\leq b\leq 4}{\underset{1\leq a\leq3}{\ }} \\ \displaystyle\left\{ E_b^{j-2p+4}\right\}\underset{1\leq b\leq 4}{\ } \end{array}\right]\label{ItzZub-sec6.f}
\end{eqnarray}
with
\begin{equation}
D_3=\sum_{1\leq a<b\leq 4}\left[E_a^2E_b^2\frac{\partial^2}{\partial E_a\partial E_b}-\frac{1}{2}\frac{E_aE_b}{E_a-E_b}\left(E_a^2\frac{\partial}{\partial E_a}-E_b^2\frac{\partial}{\partial E_b}\right)\right]
\end{equation}
which is our main result of this section. This result is more explicit than the ones of other approaches, cf. Refs.~\cite{GuhrKohler-b,Heiner-thesis,MR1893696,BergereEynard}, and hopefully may contribute to the discussion of the transition between Poisson and GOE statistics in random matrix theory. Hereby we underline that the term on which the differential operator $D_3^p$ acts is a Schur polynomial and the operator $D_3$ is essentially the Laplace operator for the case of diagonalized real symmetric random matrices. Therefore both objects are well-known to random matrix theorists and many properties of them are known.

\section{Conclusions}\label{conclusio}

We generalized the idea of Pizzetti's formula, originally stated for integrations over spheres~\cite{Pizzetti}, to integrations over Stiefel manifolds. This formula rewrites the integral of a function as an action of a differential operator followed by evaluation in the origin of the same function.  Thereby we dealt with the real case ($\beta=1$, $\St^{(1)}(n,m)=\SO(n)/\SO(m)$), the complex case ($\beta=2$, $\St^{(2)}(n,m)=\rmU(n)/\rmU(m)$), and the quaternion case ($\beta=4$, $\St^{(4)}(n,m)=\rmUSp(n)/\rmUSp(m)$) in a unifying way. The special case $m=n-1$ recovers the classical Pizzetti formula over the unit sphere.

We found:
\begin{itemize}
\item[1)] A formula in terms of the traces of the powers of the multiplication of the gradient with its conjugate for $\beta$ equal to $2$ or $4$, summarized in subsection \ref{newPiz}.
\item[2)] A very compact and explicit formula in terms of two such traces for $m= n-2$ and all $\beta$.  These formulae were rigorously proven in theorems~\ref{Pizzettinew}, \ref{stiefelbeta2}, and \ref{stiefelbeta4}.
\item[3)] A recursion of the differential operators, which holds for arbitrary $n$, $m$, $\beta$ and is summarized in theorem~\ref{recursion-theorem}.
\item[4)] An interesting Howe dual pair associated to each Stiefel manifold and an alternative interpretation of our integral formulae in that context.
\end{itemize} 
The recursions mentioned in point 3) are reminiscent to those recursions found in Refs.~\cite{MR1893696,BergereEynard} since the main idea is similar in all these approaches by splitting off the columns of the group elements in the fundamental representation one by one which are essentially integrations over spheres. In this way we applied the original Pizzetti formula, recursively.

We applied the Pizzetti  formula in point 2) for the case $\St^{(1)}(4,2)$ to an Itzykson-Zuber integral frequently appearing in random matrix theory. Thereby we found an expression more compact than the already known expressions which were derived by recursions~\cite{Heiner-thesis,GuhrKohler-b}. We hope that this expression might help in solving the tremendously complicated calculation of the two-point correlation function of the transition between Poisson and GOE statistics, see Refs.~\cite{GM-GW,Haakebook}. We underline that our result can be certainly improved since it still depends on the Laplace operator known for diagonalized real symmetric random matrices~\cite{ItzZub,MR1893696} and on the Schur polynomials.

Moreover our results may also shed some light on the explicit form of the Jack polynomials corresponding to the Dyson index $\beta=1$ and $\beta=4$. Up to now, only recursive formulae are known of these polynomials, see Refs.~\cite{Muirhead,OkoOls}. Especially the algebraic structures  we derived in terms of determinants and Pfaffian determinants might be helpful for this task. We emphasize that the formulae~\eqref{Pizzetti-beta2}, \eqref{Pizzetti-beta4-ev}, and \eqref{Pizzetti-beta4-odd} [corresponding to point 1)] reflect many of the algebraic structures found for many group integrals which were successfully calculated before. Though these formulae are not rigorously proven we are nonetheless quite confident that also these conjectures hold for group integrals over polynomials and Schwartz functions.

\section*{Acknowledgements}
KC is a Post-doctoral Fellow of the Research Foundation - Flanders (FWO). MK acknowledges financial support by the Alexander-von-Humboldt foundation.

\appendix

\section{Derivation of Pizzetti formula for $\St^{(4)}(m,n)$}\label{app1}

We want to simplify integral \eqref{groupint} for $\beta=4$,
\begin{eqnarray}\label{groupint-bet4.a}
 \widehat{\Psi}^{(4)}_{n,m}(B)=\frac{\int_{\Gl^{(4)}(n,n-m)}d[A]\delta(A^\dagger A-\eins_{\gamma (n-m)})\exp\left[\imath\tr B^\dagger A/2\right]}{\int_{\Gl^{(4)}(n,n-m)}d[A]\delta(A^\dagger A-\eins_{\gamma (n-m)})},
\end{eqnarray}
to obtain the Pizzetti formulae \eqref{Pizzetti-beta4-ev} and \eqref{Pizzetti-beta4-odd}. Recall from Eq.~\eqref{descimH} that the complex conjugation of $B$ is $B^*=\tau_2^{(n)}B\tau_2^{(n-m)}$. We extend the matrix $B$ to the $2n\times 2n$ square matrix $\widehat{B}=\left[\begin{array}{cc} B & 0 \end{array}\right]\in \Gl^{(4)}(n,0)$ and the integral to one over the full group $\rmUSp(2n)$ such that we have
\begin{eqnarray}\label{groupint-bet4.g}
 \widehat{\Psi}^{(4)}_{n,m}(B)=\frac{\int_{\rmUSp(2n)}d\mu(U)\exp\left[\imath\tr \widehat{B}^\dagger U/2\right]}{\int_{\rmUSp(2n)}d\mu(U)} =\widehat{\Psi}^{(4)}_{n,0}(\widehat{B}).
\end{eqnarray}
This group integral can be understood as the result of a saddlepoint approximation of a group integral over $\rmU(2n)$ via the limit
\begin{eqnarray}\label{groupint-bet4.c}
 \widehat{\Psi}^{(4)}_{n,m}(B)=\lim_{N\to\infty}\frac{\int_{\rmU(2n)}d\mu(U){\det}^{-2N}U\exp\left[N\tr U\tau_2^{(n)}U^T\tau_2^{(n)}+\imath\tr \widehat{B}^\dagger U/2\right]}{\int_{\rmU(2n)}d\mu(U){\det}^{-2N}U\exp\left[N\tr U\tau_2^{(n)}U^T\tau_2^{(n)}\right]}
\end{eqnarray}
because the saddlepoint equation is $\tau_2^{(n)}U^T\tau_2^{(n)}=U^{-1}$, see Ref.~\cite{KVZ14}.  The auxiliary variable $N$ can be chosen as a positive integer such we can apply a reversed version of the superbosonization formula \cite{Sommers,LSZ} and replace the integral over $U\in\rmU(2n)$ by an integral over a complex rectangular $2N\times2n$ matrix $V$ whose matrix elements are Grassmann variables (anti-commuting variables), only. For an introduction to superanalysis we refer to the textbook by Berezin~\cite{Berezin}.

The integral reads now
\begin{eqnarray}\label{groupint-bet4.h}
 \widehat{\Psi}^{(4)}_{n,m}(B)=\lim_{N\to\infty}\frac{\int d[V]\exp\left[-N\tr V^\dagger V\tau_2^{(n)}V^TV^*\tau_2^{(n)}+\imath\tr \widehat{B}^\dagger V^\dagger V/2\right]}{\int d[V]\exp\left[-N\tr V^\dagger V\tau_2^{(n)}V^TV^*\tau_2^{(n)}\right]}.
\end{eqnarray}
We linearize the quartic terms in $V$ by an auxiliary complex symmetric matrix $H\in{\rm Sym}_{\mathbb{C}}(2N)=\{K\in\mathbb{C}^{2N\times2N}|K=K^T\}$ such that the expression for $\widehat{\Psi}^{(4)}_{n,m}(B)$ becomes 
\begin{equation}
\lim_{N\to\infty}\frac{\int d[V]\int_{{\rm Sym}_{\mathbb{C}}(2N)} d[H]\exp\left[-\tr H^\dagger H/N+\tr HV\tau_2^{(n)}V^T+\tr H^\dagger V^*\tau_2^{(n)}V^\dagger+\imath\tr \widehat{B}^\dagger V^\dagger V/2\right]}{\int d[V]\int_{{\rm Sym}_{\mathbb{C}}(2N)} d[H]\exp\left[-\tr H^\dagger H/N+\tr HV\tau_2^{(n)}V^T+\tr H^\dagger V^*\tau_2^{(n)}V^\dagger\right]}.
 \label{groupint-bet4.e}
\end{equation}
The integral over $V$ yields a Pfaffian,
\begin{eqnarray}\label{int-Pfaff}
 &&\int d[V]\exp\left[\tr HV\tau_2^{(n)}V^T+\tr H^\dagger V^*\tau_2^{(n)}V^\dagger+\imath\tr \widehat{B}^\dagger V^\dagger V/2\right]\\
 &\propto& \Pf\left[\begin{array}{cc} H\otimes\tau_2^{(n)} & \displaystyle -\frac{\imath}{4}\eins_{2N}\otimes  \widehat{B}^\dagger \\ \displaystyle \frac{\imath}{4}\eins_{2N}\otimes \tau_2^{(n)}  \widehat{B}\tau_2^{(n)}   & H^\dagger\otimes\tau_2^{(n)}  \end{array}\right]
\propto \Pf\left[H^\dagger H\otimes\tau_2^{(n)}-\frac{1}{16}\eins_{2N}\otimes \widehat{B}^\dagger\widehat{B} \tau_2^{(n)}\right].\nonumber
\end{eqnarray}
Notice that this expression makes it obvious that the integral~\eqref{groupint-bet4.a} only depends on the singular values $\Lambda={\rm diag}(\Lambda_1,\ldots,\Lambda_{n-m})$ of $B$, which are all Kramer's degenerate.

After diagonalizing $H=\widetilde{U}E\widetilde{U}^T$ with $E={\rm diag}(E_1,\ldots,E_{2N})\in\mathbb{R}_+^{2N}$ and $\widetilde{U}\in\rmU(2N)$ and integrating over $\widetilde{U}$, the integral~\eqref{groupint-bet4.a} is equal to a partition function and reads
\begin{eqnarray}
 \widehat{\Psi}^{(4)}_{n,m}(B)&=&\lim_{N\to\infty}\frac{\int_{\mathbb{R}_+^{2N}} d[E]|\Delta_{2N}(E^2)|{\det}^{2m}E\exp\left[-\tr E^2/N\right]\prod_{j=1}^{n-m}\det(E^2-\Lambda_j^2/16\eins_{2N})}{\int_{\mathbb{R}_+^{2N}} d[E]|\Delta_{2N}(E^2)|{\det}^{2n}E\exp\left[-\tr E^2/N\right]}.\label{groupint-bet4.f}
\end{eqnarray}

\textbf{For $\mathbf{n-m}$ even}, this kind of partition function is equal to the one of the real Laguerre ensemble~\cite{Mehtabook} and its result is well-known in terms of a Pfaffian,
\begin{eqnarray}
\widehat{\Psi}^{(4)}_{n,m}(B)&=&\frac{1}{\Delta_{n-m}(\Lambda^2/16)}\lim_{N\to\infty}\frac{I_{2N+n-m,m}}{I_{2N,n}}\left(\frac{I_{2N+n-m-2,m+2}}{I_{2N+n-m,m}}\right)^{(n-m)/2}N^{(n-m)(n-m-2)/2}\nonumber\\
&&\times\Pf\left[\frac{\Lambda_a^2-\Lambda_b^2}{16}\widetilde{\Psi}^{(4,2N+n-m-2)}_{m+2,m}\left(\Lambda_a^2,\Lambda_b^2\right)\right]_{1\leq a,b\leq n-m},\label{Pfaffeven-N}
\end{eqnarray}
where we used the finite-$N$ partition function
\begin{eqnarray}\label{partfunc-N}
&&\widetilde{\Psi}^{(4,2N+n-m-2)}_{m+2,m}\left(\Lambda_a^2,\Lambda_b^2\right)\\
&:=&\frac{\int_{\mathbb{R}_+^{2N+n-m-2}} d[E]|\Delta_{2N+n-m-2}(E^2)|{\det}^{2m}E\exp\left[-\tr E^2/N\right]\prod_{j=a,b}\det(E^2-\Lambda_j^2/16\eins_{2N+n-m-2})}{\int_{\mathbb{R}_+^{2N+n-m-2}} d[E]|\Delta_{2N+n-m-2}(E^2)|{\det}^{2m+4}E\exp\left[-\tr E^2/N\right]}\nonumber
\end{eqnarray}
and the Selberg-integral
\begin{eqnarray}\label{partfuncconst-N}
I_{l,n}=\int_{\mathbb{R}_+^{l}} d[E]|\Delta_{l}(E^2)|{\det}^{2n}E\exp\left[-\tr E^2\right]=\left(\frac{4}{\pi}\right)^{l/2}\prod\limits_{j=1}^{l}\Gamma\left[1+\frac{j}{2}\right]\Gamma\left[n+\frac{j}{2}\right],
\end{eqnarray}
see~\cite{Mehtabook}. The limit $N\to \infty$ yields
\begin{eqnarray}\label{Pfaffeven-N.b}
\widehat{\Psi}^{(4)}_{n,m}(B)=\left(\prod\limits_{j=1}^{n-m}\frac{2^{-2j-1}\Gamma(2j+2m-1)}{\sqrt{\Gamma(2m+3)\Gamma(2m+1)}}\right)\frac{1}{\Delta_{n-m}(\Lambda^2/16)}\Pf\left[\frac{\Lambda_a^2-\Lambda_b^2}{16}\widetilde{\Psi}^{(4)}_{m+2,m}\left(\Lambda_a^2,\Lambda_b^2\right)\right]_{1\leq a,b\leq n-m},
\end{eqnarray}
with $\widetilde{\Psi}^{(4)}_{m+2,m}$ as introduced in Eq.~\eqref{phitilde}.

\textbf{For $\mathbf{n-m}$ odd}, we add an additional $\Lambda_{0}^2$ (after which we take its limit to infinity) in the finite-$N$ expression~\eqref{groupint-bet4.f}, in order to have an even number,
\begin{eqnarray}
 \widehat{\Psi}^{(4)}_{n,m}(B)=\lim_{N\to\infty}\lim_{\Lambda_{0}\to\infty}\frac{\int_{\mathbb{R}_+^{2N}} d[E]|\Delta_{2N}(E^2)|{\det}^{2m}E\exp\left[-\tr E^2/N\right]\prod_{j=0}^{n-m}\det(E^2-\Lambda_j^2/16\eins_{2N})}{(\Lambda_{0}/4)^{4N}\int_{\mathbb{R}_+^{2N}} d[E]|\Delta_{2N}(E^2)|{\det}^{2n}E\exp\left[-\tr E^2/N\right]}.\label{groupint-bet4.odd}
\end{eqnarray}
We underline that the order of the limits is crucial and cannot be switched. Now we can apply the intermediate result~\eqref{Pfaffeven-N} and take the limit $\Lambda_{0}\to\infty$, such that we find $\widehat{\Psi}^{(4)}_{n,m}(B)=$
\begin{eqnarray}
&&\frac{1}{\Delta_{n-m}(\Lambda^2/16)}\lim_{N\to\infty}\lim_{\Lambda_{0}\to\infty}\frac{1}{(\Lambda_{0}/4)^{4N}\det(\Lambda_0^2/16\eins_{n-m}-\Lambda^2/16)}\left(\frac{I_{2N+n-m-1,m+2}}{I_{2N+n-m+1,m}}\right)^{(n-m+1)/2}\nonumber\\
&&\times \frac{I_{2N+n-m+1,m}}{I_{2N,n}}N^{((n-m)^2-1)/2}\Pf\left[\frac{\Lambda_a^2-\Lambda_b^2}{16}\widetilde{\Psi}^{(4,2N+n-m-1)}_{m+2,m}\left(\Lambda_a^2,\Lambda_b^2\right)\right]_{0\leq a,b\leq n-m}\nonumber\\
&=&\frac{1}{\Delta_{n-m}(\Lambda^2/16)}\lim_{N\to\infty}\left(\frac{I_{2N+n-m-1,m+2}}{I_{2N+n-m+1,m}}\right)^{(n-m-1)/2}\frac{I_{2N+n-m-1,m+1}}{I_{2N,n}}N^{(n-m-1)^2/2}\label{Pfaffodd-N}\\
&&\times \Pf\left[\begin{array}{cc} 0 & \displaystyle\left\{\widetilde{\Psi}^{(4,2N+n-m-1)}_{m+1,m}\left(\Lambda_b^2\right)\right\}_{1\leq b\leq n-m} \nonumber\\ \displaystyle\left\{-\widetilde{\Psi}^{(4,2N+n-m-1)}_{m+1,m}\left(\Lambda_a^2\right)\right\}_{1\leq a\leq n-m} & \displaystyle\left\{\frac{\Lambda_a^2-\Lambda_b^2}{16}\widetilde{\Psi}^{(4,2N+n-m-1)}_{m+2,m}\left(\Lambda_a^2,\Lambda_b^2\right)\right\}_{1\leq a,b\leq n-m} \end{array}\right]
\end{eqnarray}
with
\begin{eqnarray}\label{partfunc-N.b}
&&\widetilde{\Psi}^{(4,2N+n-m-1)}_{m+1,m}\left(\Lambda_a^2\right)\\
&:=&\frac{\int_{\mathbb{R}_+^{2N+n-m-1}} d[E]|\Delta_{2N+n-m-1}(E^2)|{\det}^{2m}E\exp\left[-\tr E^2/N\right]\det(E^2-\Lambda_a^2/16\eins_{2N+n-m-2})}{\int_{\mathbb{R}_+^{2N+n-m-1}} d[E]|\Delta_{2N+n-m-1}(E^2)|{\det}^{2m+2}E\exp\left[-\tr E^2/N\right]}.\nonumber
\end{eqnarray}
Therefore we end up with
\begin{eqnarray}
\widehat{\Psi}^{(4)}_{n,m}(B)&=&\left(\prod\limits_{j=1}^{n-m-1}\frac{2^{-2j+1}\Gamma(2j+2m+1)}{\sqrt{\Gamma(2m+3)\Gamma(2m+1)}}\right)\frac{1}{\Delta_{n-m}(\Lambda^2/16)}\nonumber\\
&&\times \Pf\left[\begin{array}{cc} 0 & \displaystyle\left\{\widetilde{\Psi}^{(4)}_{m+1,m}\left(\Lambda_b^2\right)\right\}_{1\leq b\leq n-m} \\ \displaystyle\left\{-\widetilde{\Psi}^{(4)}_{m+1,m}\left(\Lambda_a^2\right)\right\}_{1\leq a\leq n-m} & \displaystyle\left\{\frac{\Lambda_a^2-\Lambda_b^2}{16}\widetilde{\Psi}^{(4)}_{m+2,m}\left(\Lambda_a^2,\Lambda_b^2\right)\right\}_{1\leq a,b\leq n-m} \end{array}\right].\label{Pfaffodd-N.b}
\end{eqnarray}
Note that by Eq. \eqref{connPizz2}, $\widetilde{\Psi}^{(4)}_{m+1,m}=\Psi_{2m+1}$, but we still have to calculate $\widetilde{\Psi}^{(4)}_{m+2,m}$, which will be done at the end of this appendix.

The diagonal $(n-m)\times(n-m)$ matrix $\Lambda$ comprises all singular values of the $2(n-m)\times2(n-m)$ Hermitian self-dual matrix $B^\dagger B$. We rewrite the function $\widehat{\Psi}^{(4)}_{n,m}$ in terms of the full matrix $B$ by extending the ratio with an additional Vandermonde determinant. Now we employ the identity $\det A \Pf C=\Pf A^T CA$ for an arbitrary square matrix $A$ and antisymmetric matrix $C$. Then, \textbf{for $\mathbf{n-m}$ even}, we have
\begin{eqnarray}\label{groupintsymplectic-even.b}
 \widehat{\Psi}^{(4)}_{n,m}(B)=\left(\prod\limits_{j=1}^{n-m}\frac{2^{2j-5}\Gamma(2j+2m-1)}{\sqrt{\Gamma(2m+3)\Gamma(2m+1)}}\right)\frac{\Pf\left[\tr G_{ab}^{(m)}(B^\dagger B)\right]_{1\leq a,b\leq n-m}}{\det\left[\tr(B^\dagger B)^{(a+b-2)}\right]_{1\leq a,b\leq n-m}}
\end{eqnarray}
and \textbf{for $\mathbf{n-m}$ odd} we have
\begin{eqnarray}\label{groupintsymplectic-odd.b}
 \widehat{\Psi}^{(4)}_{n,m}(B)&=&\left(\prod\limits_{j=1}^{n-m-1}\frac{2^{2j+1}\Gamma(2j+2m+1)}{\sqrt{\Gamma(2m+3)\Gamma(2m+1)}}\right)\\
 &&\times\frac{\Pf\left[\begin{array}{cc} 0 & \left\{\tr g_{b}^{(m)}(B^\dagger B)\right\}_{1\leq b\leq n-m} \\ \left\{-\tr g_{a}^{(m)}(B^\dagger B)\right\}_{1\leq a\leq n-m} & \left\{\tr G_{ab}^{(m)}(B^\dagger B)\right\}_{1\leq a,b\leq n-m} \end{array}\right]}{\det\left[\tr(B^\dagger B)^{(a+b-2)}\right]_{1\leq a,b\leq n-m}},\nonumber
\end{eqnarray}
where we used two matrix valued functions defined as in Eq.~\eqref{aux-func}.
We underline that the function $G_{ab}^{(m)}$  acts in the tensor space $\mathbb{C}^{2(n-m)\times2(n-m)}\otimes\mathbb{C}^{2(n-m)\times2(n-m)}$ such that the trace $\tr G_{ab}^{(m)}(B^\dagger B)$ is the trace in this tensor space.

Finally we show that the function $\widetilde{\Psi}^{(4)}_{m+2,m}$ can be expressed in terms of the Bessel function $\Psi_\nu$. For this reason we define $\Lambda_{ab}={\rm diag}(\Lambda_a,\Lambda_a,\Lambda_b,\Lambda_b)$. The function $\widetilde{\Psi}^{(4)}_{m+2,m}$ is given by 
\begin{eqnarray}\label{def-psi-tilde}
 \widetilde{\Psi}^{(4)}_{m+2,m}(\Lambda_a,\Lambda_b)&=&\frac{\int_{\Gl^{(4)}(m+2,2)}d[A]\delta(A^\dagger A-\eins_{4})\exp\left[\imath\tr A\Lambda_{ab}/2\right]}{\int_{\Gl^{(4)}(m,2)}d[A]\delta(A^\dagger A-\eins_{4})}\\
 &=&\frac{\int_{\Herm^{(4)}(2)} d[H]\exp\left[\tr(\eins_4-\imath H)/2-\tr \Lambda_{ab}^2(\eins_4-\imath H)^{-1}/8\right]{\det}^{-m-2}(\eins_4-\imath H)}{\int_{\Herm^{(4)}(2)} d[H]\exp\left[\tr(\eins_4-\imath H)/2\right]{\det}^{-m-2}(\eins_4-\imath H)}.\nonumber
\end{eqnarray}
The second line is the intermediate result~\eqref{groupint-bet4.b}. The diagonalization of the $4\times4$ Hermitian self-dual matrix $H=U{\rm diag}(E_1,E_1,E_2,E_2)U^\dagger$ yields an Itzykson~Zuber integral which is well known \cite{MR1893696}, i.e.
\begin{eqnarray}\label{def-psi-tilde.b}
 \widetilde{\Psi}^{(4)}_{m+2,m}(\Lambda_a,\Lambda_b)&\propto&\int_{\mathbb{R}^2} dE_1dE_2 (E_1-E_2)(1-\imath E_1)^{-2(m+1)}(1-\imath E_2)^{-2(m+1)}\exp\left[2-\imath (E_1+E_2)\right]\\
 &&\hspace*{-2cm}\times \left[\frac{E_1-E_2}{(\Lambda_a^2-\Lambda_b^2)^2}\left(\exp\left[-\frac{\Lambda_a^2}{4(1-\imath E_1)}-\frac{\Lambda_b^2}{4(1-\imath E_2)}\right]+\exp\left[-\frac{\Lambda_b^2}{4(1-\imath E_1)}-\frac{\Lambda_a^2}{4(1-\imath E_2)}\right]\right)\right.\nonumber\\
 &&\hspace*{-2cm}\left.-8\imath\frac{(1-\imath E_1)(1-\imath E_2)}{(\Lambda_a^2-\Lambda_b^2)^3}\left(\exp\left[-\frac{\Lambda_a^2}{4(1-\imath E_1)}-\frac{\Lambda_b^2}{4(1-\imath E_2)}\right]-\exp\left[-\frac{\Lambda_b^2}{4(1-\imath E_1)}-\frac{\Lambda_a^2}{4(1-\imath E_2)}\right]\right)\right]\nonumber\\
 &\propto&\frac{1}{(\Lambda_a^2-\Lambda_b^2)^2}\left(\frac{\Psi_{2m-1}(\Lambda_a)\Psi_{2m+1}(\Lambda_b)}{(2m-1)!(2m+1)!}-\frac{2\Psi_{2m}(\Lambda_a)\Psi_{2m}(\Lambda_b)}{(2m)!(2m)!}+\frac{\Psi_{2m+1}(\Lambda_a)\Psi_{2m-1}(\Lambda_b)}{(2m+1)!(2m-1)!}\right)\nonumber\\
 &&-\frac{8}{(\Lambda_a^2-\Lambda_b^2)^3}\left(\frac{\Psi_{2m-1}(\Lambda_a)\Psi_{2m}(\Lambda_b)}{(2m-1)!(2m)!}-\frac{\Psi_{2m}(\Lambda_a)\Psi_{2m-1}(\Lambda_b)}{(2m)!(2m-1)!}\right)\nonumber.
\end{eqnarray}
In the last line we employed an integral representation of the rescaled Bessel function~\eqref{renormBessel}.

\section{Derivation of the Pizzetti formula for $\St^{(\beta)}(n,n-2)$}\label{app0}

In this appendix we demonstrate another approach to calculate the Pizzetti formulae for all three Dyson indices in a unifying way, specific to the case $m=n-2$. The idea is to rewrite the Dirac delta function in Eq.~\eqref{groupint} as the Fourier transform of a constant function. 

Concretely, for $\beta=2$ we consider 
\begin{equation}
\label{deltaasFT}
\delta(A^\dagger A-\eins_2)\propto \int_{\Herm^{(2)}(2)} d[H]\exp\left[-\tr(A^\dagger A-\eins_{2})(\eins_2-\imath H)\right],
\end{equation}
where the constant shift in the matrix $H$ is introduced to guarantee the absolute integrability of the integral over $A$ when used in Eq.~\eqref{groupint}. Note that we suppress the normalizations temporarily in order to re-introduce it later on. Using Eq.~\eqref{deltaasFT} in Eq.~\eqref{groupint} yields
\begin{eqnarray}\label{groupint-bet2.b}
 \widehat{\Psi}^{(2)}_{n,n-2}(B)\propto\int_{\Herm^{(2)}(2)} d[H]\exp\left[\tr(\eins_2-\imath H)-\frac{1}{4}\tr B^\dagger B(\eins_2-\imath H)^{-1}\right]{\det}^{-n}(\eins_2-\imath H).
\end{eqnarray}
Now we  can expand the exponential function incorporating the inverse matrix $(\eins_2-\imath H)^{-1}$. Then we employ the property that for arbitrary invertible $M\in\mC^{2\times 2}$, we have $\tr M=\tr(M^{-1})\det M$.
In our case this relation reads
\begin{eqnarray}\label{2-d-rel}
 \tr B^\dagger B(\eins_2-\imath H)^{-1}=\frac{\det B^\dagger B}{\det(\eins_2-\imath H)}\tr (B^\dagger B)^{-1}(\eins_2-\imath H).
\end{eqnarray}
Notice that we assume that $B^\dagger B$ is invertible. This assumption is not a restriction because the invertible dyadic matrices are dense in the set of all dyadic matrices since $n\geq2$ and the non-invertible matrices are a set of measure zero. Hence $\widehat{\Psi}^{(2)}_{n,n-2}(B)$ is proportional to
$$ \sum_{j=0}^\infty\frac{(-1)^j}{4^j}{\det}^jB^\dagger B\frac{1}{j!}\left.\frac{\partial^j}{\partial\mu^j}\right|_{\mu=0}\int_{\Herm^{(2)}(2)} d[H]\exp\left[\tr(\eins_2-\imath H)(\eins_2+\mu (B^\dagger B)^{-1})\right]{\det}^{-n-j}(\eins_2-\imath H)$$ \vspace{-4mm}
\begin{eqnarray}
\hspace{-14mm}\propto\sum_{j=0}^\infty\frac{(-1)^j}{4^j\Gamma(n+j)\Gamma(n+j-1)}{\det}^{2-n}B^\dagger B\frac{1}{j!}\left.\frac{\partial^j}{\partial\mu^j}\right|_{\mu=0}{\det}^{n+j-2}(B^\dagger B+\mu\eins_2),\label{groupint-bet2.c}
\end{eqnarray}
where we already integrated over $H$ in the last line by performing an Ingham-Siegel integral \cite{Ingham,Siegel}. Expanding the determinant we have
\begin{eqnarray}\label{groupint-bet2.d}
 \widehat{\Psi}^{(2)}_{n,n-2}(B) &=&\sum_{j=0}^\infty\sum_{l=0}^{\lfloor j/2\rfloor}\frac{(-1)^j\Gamma(n)\Gamma(n-1)}{4^j\Gamma(n+j)\Gamma(n-1+l)(j-2l)!l!}\tr^{j-2l}B^\dagger B{\det}^{l}B^\dagger B,
\end{eqnarray}
which yields Eq.~\eqref{Pizzetti-beta2-2bb} for $\beta=2$ after replacing $B$ by the gradient in the matrix $A$. 

This calculation can be readily extended to the real case $\beta=1$. Now we introduce an auxiliary matrix $H\in\mR^{2\times2}$ and find
\begin{eqnarray}\label{groupint-bet1.b}
 \widehat{\Psi}^{(1)}_{n,n-2}(B)\propto\int_{\Herm^{(1)}(2)} d[H]\exp\left[\tr(\eins_2-\imath H)-\frac{1}{4}\tr B^T B(\eins_2-\imath H)^{-1}\right]{\det}^{-n/2}(\eins_2-\imath H).
\end{eqnarray}
Again we expand the exponential function in $\tr B^T B(\eins_2-\imath H)^{-1}$ and apply relation~\eqref{2-d-rel}. As the analogue of the result in~\eqref{groupint-bet2.d} we therefore obtain
\begin{eqnarray}\label{groupint-bet1.d}
 \widehat{\Psi}^{(1)}_{n,n-2}(B) &=&\sum_{j=0}^\infty\sum_{l=0}^{\lfloor j/2\rfloor}\frac{(-1)^j\Gamma(n/2)\Gamma((n-1)/2)}{4^j\Gamma(n/2+j)\Gamma((n-1)/2+l)(j-2l)!l!}\tr^{j-2l}B^T B\,\,{\det}^{l}B^T B,
\end{eqnarray}
after performing the Ingham-Siegel integral \cite{Ingham,Siegel}. Replacing $B$ by the gradient in $A$ we find the result~\eqref{Pizzetti-beta2-2bb} for $\beta=1$.

For $\beta=4$, we take $H$ a Hermitian self-dual $4\times4$ matrix. Then the integral~\eqref{groupint} becomes
\begin{eqnarray}\label{groupint-bet4.b}
 \widehat{\Psi}^{(4)}_{n,n-2}(B)\propto\int_{\Herm^{(4)}(2)} d[H]\exp\left[\frac{1}{2}\tr(\eins_4-\imath H)-\frac{1}{8}\tr B^\dagger B(\eins_4-\imath H)^{-1}\right]{\det}^{-n}(\eins_4-\imath H).
\end{eqnarray}
When expanding the exponential in the term $\tr B^\dagger B(\eins_4-\imath H)^{-1}$, we have to modify the relation~\eqref{2-d-rel}, since the inverse of a $4\times4$ matrix of the form
\begin{eqnarray}\label{matele4}
M=\left[\begin{array}{cccc} a & 0 & z_1 & z_2 \\ 0 & a & -z_2^* & z_1^* \\ z_1^* & -z_2 & b & 0 \\ z_2^* & z_1 & 0 & b \end{array}\right]
\end{eqnarray}
is
\begin{eqnarray}\label{matinv4}
M^{-1}=\frac{1}{ab-|z_1|^2-|z_2|^2}\left[\begin{array}{cccc} b & 0 & -z_1 & -z_2 \\ 0 & b & z_2^* & -z_1^* \\ -z_1^* & z_2 & a & 0 \\ -z_2^* & -z_1 & 0 & a \end{array}\right].
\end{eqnarray}
Thus we find $\tr M=\Pf(IM)\,\tr M^{-1}$, so
\begin{eqnarray}\label{2-d-rel-b}
 \tr B^\dagger B(\eins_2-\imath H)^{-1}=\frac{\Pf I B^\dagger B}{\Pf I (\eins_2-\imath H)}\tr (B^\dagger B)^{-1}(\eins_4-\imath H),
\end{eqnarray}
with the definitions of the constant matrix $I$, see Eq.~\eqref{I-def},  and the Pfaffian~\eqref{Pfaffiandef}. After performing the Ingham-Siegel integral for $\beta=4$ over the matrix $H$, we encounter the Pfaffian $\Pf I(\eins_4+\mu (B^\dagger B)^{-1})=1+\mu\tr(B^\dagger B)^{-1}/2+\mu^2\Pf I(B^\dagger B)^{-1}$. The result is 
\begin{eqnarray}\label{groupint-bet4.d}
 \widehat{\Psi}^{(4)}_{n,n-2}(B) &=&\sum_{j=0}^\infty\sum_{l=0}^{\lfloor j/2\rfloor}\frac{(-1)^j\Gamma(2n)\Gamma(2(n-1))}{4^j\Gamma(2n+j)\Gamma(2(n-1)+l)(j-2l)!l!}\left(\frac{\tr B^\dagger B}{2}\right)^{j-2l}{\Pf}^{l}(IB^\dagger B),
\end{eqnarray}
which becomes Eq.~\eqref{Pizzetti-beta4-2b}. The factor $1/2$ in front of the trace $\tr B^\dagger B$ normalizes the terms correctly since the singular values are Kramers degenerate.

\section{An abstract type of invariant integral}
\label{secprop}

In this section we consider two variables $u,v\in\mR^{km}$ for $k,m\in\mN$, along with a set of $k$ real $km\times km$-matrices $\{J^{(0)},J^{(1)},\cdots, J^{(k-1)}\}$, which satisfy the following requirement:
\begin{equation}
\left(J^{(i)}\right)^T J^{(j)}+ \left(J^{(j)}\right)^T J^{(i)}=2\delta_{ij}\eins_{km},\qquad \mbox{for all}\quad i,j\in\{0,\cdots,k-1\}.
\end{equation}
In particular this implies that $J^{(j)}\in \SO(km)$ for $1\le j\le k$. Using the notation of the preliminaries we define two differential operators on $\mR^{km\times 2}$ as
\begin{equation}
\label{defAB}
A:= \Delta_u+\Delta_v\quad \mbox{and }\quad B:=\Delta_u\Delta_v-\sum_{j=0}^{k-1} \langle \nabla_u,J^{(j)}\nabla_v\rangle ^2.
\end{equation}

With this purpose in mind we consider the orthogonal transformation $u'=J^{(0)}u$, which shows that we could also have chosen $\{\eins_{km},\left(J^{(0)}\right)^TJ^{(1)},\cdots, \left(J^{(0)}\right)^TJ^{(k-1)}\}$ as an original choice of matrices. Thus we assume $J^{(0)}=\eins_{km}$ in the ensuing discussion without loss of generality, which implies that the $k-1$ remaining matrices $\{J^{(1)},\cdots, J^{(k-1)}\}$ satisfy 
\begin{equation}
\label{Clifford}
\begin{cases}
(J^{(i)})^T=-J^{(i)},&\mbox{for }i>0;\\
J^{(i)}J^{(j)}+J^{(j)}J^{(i)}=-2\delta_{ij}\eins_{km},&\mbox{for }\quad i,j\in\{1,\cdots, k-1\}.\end{cases}
\end{equation}

\begin{remark}
\label{remCli}
We note that the second condition in Eq. \eqref{Clifford} yields a Clifford algebra structure. Hence the definition of $k-1$ of such matrices can alternatively be described as an algebra morphism $\mR_{0,k-1}\to \End( \mR^{km})$ from the Clifford algebra $\mR_{0,k-1}$, with signature $0,k-1$, to the associative algebra of real $km\times  km$-matrices, which restricts to an injective morphism from $\mR^{k-1}$ (canonically embedded in $\mR_{0,k-1}$) to the space of anti-symmetric matrices. 
\end{remark}
\noindent The three cases that are explicitly used in the current paper satisfy $J^{(0)}=\eins_{km}$ and are the real ($\mR_0$), complex ($\mR_{0,1}$) and quaternion (a quotient of $\mR_{0,3}$) numbers.

With help of the matrices $J^{(i)}$ we can prove the following lemmas, propositions and theorems for general $k$.

\begin{lemma}
\label{Bu1}
For $l\in\mN$, the operator ${B}_u^{(l)}$ defined as
\begin{equation}
{B}^{(l)}_u=4l\left(\left(\mE_u+k\frac{m-1}{2}+l-1\right)\Delta_v-\sum_{j=0}^{k-1}\langle u,J^{(j)}\nabla_v\rangle \langle \nabla_u,J^{(j)}\nabla_v\rangle\right),
\end{equation}
satisfies $[B^l,u^2]={B}_u^{(l)}B^{l-1}$, with $B$ given in Eq. \eqref{defAB}.
\end{lemma}
\begin{proof}
Application of the commutation relations in Lemma \ref{commrel} yields
\begin{equation}
[B,u^2]=4\left(\left(\mE_u+k\frac{m-1}{2}\right)\Delta_v-\sum_{j=0}^{k-1}\langle u,J^{(j)}\nabla_v\rangle \langle \nabla_u,J^{(j)}\nabla_v\rangle\right),\label{a:eq1}
\end{equation}
which equals $B_u^{(1)}$. Further application of Lemma \ref{commrel} also implies
\begin{eqnarray}
[B,B_u^{(1)}]&=&8B\Delta_v-4[B,\sum_{j=0}^{k-1}\langle u,J^{(j)}\nabla_v\rangle \langle \nabla_u,J^{(j)}\nabla_v\rangle]\nonumber\\
&=&8B\Delta_v-8\sum_{j=0}^{k-1}\Delta_v \langle \nabla_u,J^{(j)}\nabla_v\rangle^2+8\sum_{i,j=0}^{k-1}\langle\nabla_u,J^{(i)}\nabla_v\rangle\langle J^{(i)}\nabla_v,J^{(j)}\nabla_v\rangle \langle \nabla_u,J^{(j)}\nabla_v\rangle\nonumber\\
&=&8B\Delta_v+16\sum_{0\leq i<j\leq k-1}\langle \nabla_u,J^{(i)}\nabla_v\rangle\langle J^{(i)}\nabla_v,J^{(j)}\nabla_v\rangle \langle \nabla_u,J^{(j)}\nabla_v\rangle\nonumber\\
&=&8B\Delta_v,\label{a:eq2}
\end{eqnarray}
where we employed Eq.~\eqref{Clifford} in the last line. Combining the two relations~\eqref{a:eq1} and \eqref{a:eq2} yields
\begin{equation}
[B^l,u^2]=4l(l-1)\Delta_vB^{l-1}+lB_u^{(1)}B^{l-1},
\end{equation}
which proves the lemma.
\end{proof}

\begin{lemma}
\label{AB1}
For $k,l\in\mN$, $A$ as introduced in equation \eqref{defAB} and $B^{(l)}_u$ as introduced in lemma \ref{Bu1}, the relation $[A^k,{B}^{(l)}_u]=8klA^{k-1}B$ holds.
\end{lemma}
\begin{proof}
The equations in lemma~\ref{commrel} imply $[A,B^{(l)}_u]=8lB$. Since $A$ and $B$ commute, the proposed result follows immediately.
\end{proof}

\begin{proposition}
\label{prop}
The functional $T$ on polynomials on $\mR^{km\times 2}$ defined as
\begin{equation}
\label{eqprop}
T(f)=\sum_{j=0}^\infty\frac{\Gamma(km/2)}{4^j\Gamma(j+km/2)}\sum_{l=0}^{\lfloor j/2\rfloor}\frac{\Gamma\left(k(m-1)/2\right)}{\Gamma\left(l+k(m-1)/2\right)}\left(\frac{A^{j-2l}}{(j-2l)!}\frac{B^l}{l!}f\right)(0),\end{equation}
with operators $A$ and $B$ as introduced in equation \eqref{defAB}, satisfies $T(u^2f)=T(f)$.
\end{proposition}

\begin{proof}
Without loss of generality we assume $\mE_uf=2af$ and $\mE_v f=2bf$. It is clear that the numbers $2a$ and $2b$ have to be even in order to obtain non-zero outcomes, since $A$ and $B$ preserve the parity of the degree in $u$ and $v$ individually. Then we calculate
\begin{eqnarray}
T(u^2f)=\frac{4^{-a-b-1}\Gamma(km/2)}{\Gamma(km/2+a+b+1)}\sum_{l=0}^{\lfloor (a+b+1)/2\rfloor}\frac{\Gamma\left(k(m-1)/2\right)}{\Gamma\left(k(m-1)/2+l\right)}\left(\frac{A^{a+b+1-2l}}{(a+b+1-2l)!}\frac{B^l}{l!}u^2f\right)(0).
\end{eqnarray}
Applying lemma \ref{Bu1} and lemma \ref{AB1} yields
\begin{eqnarray}
\frac{\Gamma(km/2+a+b+1)}{\Gamma(km/2)4^{-(a+b+1)}}T(u^2f)&=&4\sum_{l=1}^{\lfloor (a+b)/2\rfloor}\frac{\Gamma\left(k(m-1)/2\right)2l}{\Gamma\left(k(m-1)/2+l\right)}\left(\frac{A^{a+b-2l}}{(a+b-2l)!}\frac{B^l}{l!}f\right)(0)\nonumber\\
&+&\sum_{l=1}^{\lfloor (a+b+1)/2\rfloor}\frac{\Gamma\left(k(m-1)/2\right)}{\Gamma\left(k(m-1)/2+l\right)}\left(\frac{B^{(l)}_uA^{a+b+1-2l}}{(a+b+1-2l)!}\frac{B^{l-1}}{l!}f\right)(0)\nonumber\\
&+&\sum_{l=0}^{\lfloor (a+b)/2\rfloor}\frac{\Gamma\left(k(m-1)/2\right)}{\Gamma\left(k(m-1)/2+l\right)}\left(\frac{A^{a+b+1-2l}}{(a+b+1-2l)!}u^2\frac{B^l}{l!}f\right)(0).
\end{eqnarray}

Using the properties $\mE_uf=2af$ and $\mE_v f=2bf$ and the definition of $B^{(l)}_u$ in lemma \ref{Bu1} shows that the second term on the right-hand side is equal to
\begin{equation}
4\sum_{l=1}^{\lfloor (a+b+1)/2\rfloor}\frac{\Gamma\left(k(m-1)/2\right)}{\Gamma\left(k(m-1)/2+l-1\right)}\left(\frac{\Delta_u^{a-l+1}\Delta_v^{b-l+1}}{(a-l+1)!(b-l)!}\frac{B^{l-1}}{(l-1)!}f\right)(0),
\end{equation}
while the third term is equal to
\begin{equation}
\sum_{l=0}^{\lfloor (a+b)/2\rfloor}\frac{\Gamma\left(k(m-1)/2\right)}{\Gamma\left(k(m-1)/2+l\right)}\left(\frac{\Delta_u^{a-l+1}\Delta_v^{b-l}}{(a-l+1)!(b-l)!}u^2\frac{B^l}{l!}f\right)(0).
\end{equation}
Using $[\Delta_u^{a-l+1},u^2]=4(a-l+1)(\mE_u+km/2+a-l)\Delta_u^{a-l}$ (which is a direct consequence of the first property in lemma \ref{commrel}) for this third term and adding everything up then yields
\begin{eqnarray}
T(u^2f)&=&\frac{\Gamma(km/2)}{4^{a+b}\Gamma(km/2+a+b+1)}\sum_{l=1}^{\lfloor (a+b)/2\rfloor}\frac{\Gamma\left(k(m-1)/2\right)2l}{\Gamma\left(k(m-1)/2+l\right)}\frac{A^{a+b-2l}}{(a+b-2l)!}\frac{B^l}{l!}f\\\
&+&\frac{\Gamma(km/2)}{4^{a+b}\Gamma(km/2+a+b+1)}\sum_{l=0}^{\lfloor (a+b-1)/2\rfloor}\frac{\Gamma\left(k(m-1)/2\right)(b-l)}{\Gamma\left(k(m-1)/2+l\right)}\frac{A^{a+b-2l}}{(a+b-2l)!}\frac{B^l}{l!}f\nonumber\\
&+&\frac{\Gamma(km/2)}{4^{a+b}\Gamma(km/2+a+b+1)}\sum_{l=0}^{\lfloor (a+b)/2\rfloor}\frac{\Gamma\left(k(m-1)/2\right)(km/2+a-l)}{\Gamma\left(k(m-1)/2+l\right)}\frac{A^{a+b-2l}}{(a+b-2l)!}\frac{B^l}{l!}f.\nonumber
\end{eqnarray}
The three sums would add up to $T(f)$ if the upper limit in the second term were $\lfloor (a+b)/2\rfloor$ and not $\lfloor (a+b-1)/2\rfloor$. This is only relevant when $a+b=2l$ is even, which is therefore the case we focus on. In order for $B^{a+b}f$ not to be zero we must have $a=b=l$. The factor $(b-l)$ in that term implies that we can replace the upper limit by the desired one, which concludes the proof.
\end{proof}

Finally we define two commuting Lie groups with action on $\mR^{km\times 2}$. The first is defined as
\begin{eqnarray}
G&:=&\{A\in \Gl(km;\mR)\,|\, A^TJ^{(i)}A=J^{(i)},\quad\mbox{for }0\le i\le k-1 \}\nonumber\\
&=&\{A\in {\rm O}(km;\mR)\,|\, A^TJ^{(i)}A=J^{(i)},\quad\mbox{for }1\le i\le k-1 \}.
\end{eqnarray}
The second group we define is $\SO(k+1)$, which has an action on $\mR^{km\times 2}$ as follows. We consider the $k+1$ linearly independent vectors $u, v, J^{(1)}v,\cdots ,J^{(k-1)}v$. The orthogonal group acting on the space corresponding to the span of these vectors is naturally embedded in $\Gl(\mR^{km\times 2})$.

\begin{theorem}\label{absthm}
The unique $G\times \SO(k+1)$-invariant integration on the manifold $M$, defined as the submanifold of $\mR^{km\times 2}$ with $m>1$ corresponding to the intersection of the hypersurfaces determined by the relations $u^2=1$, $v^2=1$ and $\langle u, J^{(i)}v\rangle=0$ for $i=0,\cdots,k-1$, corresponds to the functional \eqref{eqprop}.
\end{theorem}
\begin{proof}
This follows from the ideas in section \ref{secinv} and proposition \ref{prop}.
\end{proof}

Note that in this paper we consider the three cases $k=\beta\in\{1,2,4\}$, where we have the exceptional morphisms $\mathfrak{so}(k+1)\cong\mathfrak{su}(2)$ for $k=2$ and $\mathfrak{so}(k+1)\cong \mathfrak{usp}(4)$ for $k=4$. In section \ref{sec2} we prove explicitly that these situations are special cases of theorem \ref{absthm}.

\end{document}